\documentclass[journal]{IEEEtran}

\usepackage{lineno,hyperref}
\usepackage{amsmath}
\usepackage{lineno,multicol}
\usepackage{graphicx}
\usepackage{makecell}
\usepackage{subfig}
\usepackage{float}
\usepackage{epstopdf}
\usepackage{amsthm}
\usepackage{cases}
\usepackage{pbox}
\usepackage{extarrows}
\usepackage{cite}
\usepackage{amssymb}
\usepackage{mathrsfs}
\newtheorem{theorem}{Theorem}
\newtheorem{lemma}{Lemma}

\if CLASSINFOpdf
\else

\fi

\begin{document}

\title{High-resolution chirplet transform:   from parameters  analysis to  parameters combination}
\author{Xiangxiang Zhu, Bei Li,  Kunde Yang,  Zhuosheng Zhang, Wenting Li 
\thanks{This work was supported by the National Natural Science Foundation of China under Grant no. U20B2075, the Fundamental Research Funds for the Central Universities  under Grant no.  G2021KY05103.}
\thanks{X. Zhu is with School  of Mathematics and Statistics, Northwestern Polytechnical University,  Xi'an 710072, China (e-mail: zhuxiangxiang@nwpu.edu.cn)} 	
\thanks{ B. Li, Z. Zhang and  W. Li are with School  of Mathematics and Statistics, Xi'an Jiaotong University,  Xi'an, 710049, China (e-mail: libeixjtu@163.com; zszhang@mail.xjtu.edu.cn;  wenting1990@stu.xjtu.edu.cn).}
\thanks{K. Yang  is with School of Marine Science and Technology, Northwestern Polytechnical University, Xi’an 710072, China (e-mail: ykdzym@nwpu.edu.cn).}}

\markboth{IEEE Transactions on Signal Processing}%
{Shell \MakeLowercase{\textit{et al.}}: Bare Demo of IEEEtran.cls for IEEE Journals}
\maketitle

\begin{abstract}
The standard chirplet transform (CT) with a chirp-modulated Gaussian window provides a valuable tool for analyzing linear chirp signals.  The parameters present in the window determine the performance of the CT and play a vital role in high-resolution time-frequency (TF) analysis.  In this paper,   we give the window shape analysis of the CT and compare it with the extension that employs a  rotating Gaussian window by the fractional Fourier transform.  The given parameters analysis provides theoretical guidance for developing high-resolution CT. We then propose a multi-resolution chirplet transform (MrCT) by combining multiple CTs with different parameter combinations.  These are combined geometrically to obtain an improved TF resolution by overcoming the limitations of any single representation of the CT.  By deriving a  combined instantaneous frequency equation, we further develop a  high-concentration TF post-processing approach to improve the readability of the  MrCT. Numerical experiments on simulated and real signals verify the effectiveness of the proposed methods.
\end{abstract}


\IEEEpeerreviewmaketitle

\section{Introduction}

\IEEEPARstart
Multi-component non-stationary signals are widely observed in many real-world applications, such as seismic  [1,2], radar  and  sonar [3,4], biomedicine  [5,6],  and  mechanical  engineering  [7,8].   Such signals  contain several components with time-varying characteristics and often express complex oscillation content (like fast-varying or closely-spaced instantaneous frequencies (IFs)) [4-6,9].  To capture the useful  characteristics of multi-component non-stationary signals as accurately as possible is of great significance for various applications because that can provide insight into  the complex structure of the signal,  help to understand the data system and predict its future behavior. 

Time-frequency (TF) analysis methods [10] provide an effective tool for characterizing  multi-component non-stationary signals with  time-varying frequency content. In the past decades, the study of reliable TF representations for non-stationary signal analysis has been a hot research topic. The  most well-known  TF method  is the linear TF transforms, such as the short-time Fourier transform (STFT) [11], the continuous wavelet transform (CWT) [12], and  the chirplet transform (CT) [13].  In linear TF analysis,  the signal is studied via inner products with a  pre-assigned basis. The main limitation of   this kind of method is that  they  can not  precisely  localize a signal in both time and frequency simultaneously because of the effect of the Heisenberg uncertainty principle [5,10].  In order to achieve high  TF resolution,  the bilinear TF representations,  represented by the Wigner-Ville distribution (WVD) [14], are presented. Good TF resolution can be obtained by the WVD but  in addressing  multi-component or non-linear frequency modulated signals,   it suffers from cross-terms (TF artifacts),  which  renders  it unusable for some practical applications [15]. Various  smoothed versions of  WVD by introducing  kernel smoothing are developed to reduce the unwanted inferences,  which leads to the Cohen’s class or Affine class [10,16], but the kernel smoothing  again blurs the TF distribution.

To increase TF concentration and obtain good TF readability,  the post-processing methods of linear TF representations  have  been widely studied.  The reassignment method (RM)  [17,18] reallocates the TF coefficients from the original position to the center of gravity of the signal's energy distribution such that a sharpened TF representation is yielded.  As a special case of RM, the synchrosqueezing transform (SST), put forward by Daubechies and Maes in the mid-1990s [19], squeezes the TF coefficients into the instantaneous frequency  (IF)  trajectory only in the  frequency direction, which not only improves the TF concentration but also  allows for signal reconstruction  [20].  Over the last decade,  the SST has received considerable research attentions. It has been extended to different transform frameworks, including  the STFT-based SST [21], the synchrosqueezed curvelet transform [22],  the synchrosqueezing S-transform [23], the  chirplet-based SST [24], etc.  It is known that one drawback associated with SST  is that it suffers from a low TF resolution when dealing with fast-varying  signals [25,26]. To address this drawback,  many improvements have been presented involving  the demodulated SST [25,27], the high-order SST [26,28], the multiple squeezes transform [24,29], the synchroextracting chirplet transform [30], and the time-reassigned  SST [31,32]. The TF post-processing methods introduced above have been increasingly used and adapted in many fields [6,7,23,27,28,31,32]. However,  these methods have  the intrinsic limitation that they operate on a linear TF representation, associated with a fixed TF resolution  given by a global window or wavelet.

In order to improve the TF resolution more essentially,  many adaptive  TF analysis methods have been proposed. In [33], an adaptive STFT using  Gaussian window function with time-varying variance  (i.e., window width) dependent on the chirp rate (CR) of the input signal is proposed. A more powerful adaptive STFT  is developed in [34,35], where the variance used in the window function is time-frequency-varying because there exist multiple different  CRs at the same time instant for multi-component signals.   S. Pei et. al [36] also presented an improved version of adaptive STFT with a chirp-modulated  Gaussian window.  In  [37], L. Li et al. proposed   the  adaptive SST and its second-order extension based on the adaptive STFT  with a time-varying window  to further enhance the TF concentration. In addition to adaptive STFT-based methods,  the adaptive wavelet transform (AWT) [38]  is developed  by adjusting the wavelet parameter.  The  bionic wavelet transform,  proposed by Yao and Zhang [39], introduces an extra parameter to adaptively adjust the  TF resolution not only by the frequency but also by the instantaneous amplitude and its first-order differential.   The energy concentration of the S-transform has been addressed in [40,41] by optimizing the width of the window function used. Moreover, the adaptive WVD has been developed by designing different types of filters in the ambiguity domain. L. Stankovi\'{c} [42] proposed L-class WVD;  Boashash and O'Shea [43] proposed the polynomial WVD for analyzing polynomial phase signals;   Katkovnik  [44]  presented   the local polynomial  WVD; Jones and Baraniuk [45]  proposed the adaptive optimal kernel TF distribution  by introducing  a signal-dependent radially Gaussian kernel that adapts over time. Since 2013, Khan et al. conducted a series of research in developing adaptive WVD,  mainly involving  the adaptive fractional spectrogram [46], the  adaptive directional TF distribution [47], the locally optimized adaptive directional TF  distribution [48].  Despite  the high TF resolution, the adaptive TF representations, especially  the adaptive non-linear TF analysis,  generally have high computational complexity, and determining the matching parameters remains challenging [48,49].

Other high-resolution methods are based on the combination of multiple spectrograms or wavelet transforms [5,50,51].  The combined methods compute the geometric mean of multiple estimates with short and long windows or a set of wavelets, and have been shown to result in good joint TF resolution for signals consisting of mixtures of tones and pulses [5,50]. Such combinations of conventional spectrograms or  wavelet transforms do not work well, however, for signals containing mixtures of chirp-like signals; moreover,  how to improve the TF readability of the combinated transformations  is an important problem to be solved. 

The chirplet transform (CT),  proposed by Mann and Haykin [13], Milhovilovic and Bracewell [52] almost at the same time, is a popular  TF analysis method for characterizing chirp-like signals,  and it has been widely applied in different areas such as radar [53,54], biomedicine [55],  and mechanical engineering [56,57].  This method uses an  extra CR parameter than STFT (i.e., using a chirp-modulated window function, also termed chirp basis)  to characterize linear frequency variation. When the parameter equals to the chirp rate of a signal and the input window length is matching, the CT will generate a  highly concentrated  TF representation.  This fact leads to the  development of  many adaptive CTs, such as the self-tuning CT [58], the general linear CT [56], the synchro-compensating CT [59], and the double-adaptive CT [54]. The adaptability of  such methods  is achieved by modifying the window width or/and  CR parameter depending on the instantaneous signal's nature. However,  these  adaptive CTs are also  confronted with the challenges of  high computational complexity and parameter matching, as most adaptive TF analysis methods do.  To better deal with the signals with non-linear IFs, the CT has been extended to general parameterized TF transforms, which involve  the warblet transform [60,61], the polynomial CT [62], the spline-kernelled  CT [63], and the scaling-basis CT [57]. Besides,  the CT also can be seen as a three-dimensional analysis method as it represents a signal in the  joint time-frequency-CR domain when considering the  CR parameter as a variable [13,53]. In this sense, the   CT  enables separating the intersected signals with cross-over  IFs by utilizing the CR information.   The theoretical analysis and high-concentration representation of three-dimensional CT are recently introduced in  [64-66].  

Differing from the extensions of CT above, here we introduce a high-resolution analysis approach by combining multiple CTs with different window widths and CR values, which can help to achieve a high-resolution TF representation by overcoming limitations of any single representation of the CT. In this paper, we first  theoretically  analyze  the effect of  variance and CR  parameters on the TF resolution of  CT and prove that a narrow window limits the matching capacity of chirp basis. We compare the CT with its extension that employs a  rotating Gaussian window by the fractional Fourier transform [67,68]. By comparison, we  conclude that the parameters of rotation-window   CT have a clearer  geometric  meaning than that of the standard CT, but this rotation extension is actually the  CT  with a special parameter combination and thus  can not localize a signal in both time and frequency precisely.   To obtain a high-resolution TF representation,  we  propose a multi-resolution CT (MrCT) by computing  the geometric mean of  multiple CTs  with different parameter combinations, which   localizes  the  signal  in both time and frequency  better than it is possible with any single CT. Finally, we develop a TF post-processing method  to improve the  readability of MrCT based on a combined IF equation.

The remainder of the paper is organized as follows. In Section \textrm{II}, the  parameter  analysis of the standard CT and rotation-window  CT  is  presented.    In Section  \textrm{III},  we describe the  details  of the proposed MrCT method. The post-processing TF method, named multi-resolution synchroextracting chirplet transform,  is presented in Section \textrm{IV}.  Experimental results and comparative studies are introduced  in Section \textrm{V}.  Finally, the conclusions are drawn  in Section \textrm{VI}.

\section{Window shape  analysis of  CT and  its rotation-window  extension}

The following sections provide parameters analysis of  CT and compare it with a  rotation-window CT.   

\subsection{Signal model}
In this paper we consider the amplitude-modulation and frequency-modulation  (AM-FM) model  to  describe   the time-varying  features of multi-component non-stationary signals, which is given by
\begin{equation}
	f(t)=\sum_{k=1}^{K}f_k(t)=\sum_{k=1}^{K}A_k(t)\mathrm{e}^{j\phi_k(t)}, 
\end{equation}
where $K$ is a positive integer representing the number of AM-FM components, $\sqrt{-1}=j$ denotes the imaginary unit,  $A_k(t)>0$  and $\phi_k(t)$   are  the instantaneous amplitude and  instantaneous phase of the $k$-th component (or mode), respectively. The first and the second  derivatives of the  phase,  i.e.,   $\phi'_k(t)$, $\phi''_k(t)$, are referred to as  the instantaneous  frequency (IF) and the chirp rate (CR) of the $k$-th component. 
To  model  multi-component non-stationary signals  as (1) is important to  extract information hidden in $f(t)$, and this representation has been used in many applications including seismic wave analysis, medical data analysis,  mechanical vibration diagnosis and speech recognition, see for example [20,38]. 

The Fourier transform  of a given signal  $f(t)\in L^1(\mathbb{R})$, i.e. its correlation with a sinusoidal wave $\mathrm{e}^{j\omega t}$, is defined as
$$\hat{f}(\omega)=\int_{-\infty}^{+\infty}f(t)\mathrm{e}^{-j\omega t} \mathrm{d}t.$$

\subsection{Window shape analysis  of CT}
The chirplet transform (CT) [13,52] generalizes the STFT by using an extra CR parameter and its  definition is given by 
\begin{equation}
	\begin{split}
		C_{f}^{h_{\beta}}(t, \omega)=\int_{-\infty}^{+\infty}f(\mu)g^*(t-\mu)\mathrm{e}^{-j\frac{\beta}{2}(\mu- t)^2}\mathrm{e}^{-j\omega(\mu- t)}\mathrm{d}\mu,
	\end{split}
\end{equation}
where $\beta$ is the CR parameter/variable,   $g(t)$ is a window function in the Schwartz class,  superscript $*$ denotes the complex conjugate, and $h_{\beta}(t)=g(t)\mathrm{e}^{j\frac{\beta }{2}t^2}$ called  the chirp-modulated  window. 
If  $\beta=0$,  then the CT reduces to  the   STFT.  Importantly, the  squared modulus of  CT  can be interpreted as a smoothed WVD, resulting from the smoothing of the   WVD of the signal by the  WVD of the  chirp-modulated window,  as presented in Lemma 1.


\begin{lemma}
	For a  signal $f(t)$, its squared modulus of the CT can be equivalently expressed as
	\begin{equation}
		| C_{f}^{h_\beta}(t, \omega)|^2=\frac{1}{2\pi}W_f\ast\ast W_{h_{\beta}}(t, \omega),
	\end{equation}
	where $\ast\ast$ denotes the  2D convolution operator, and $W_f(t, \omega)$  and  $W_{h_{\beta}}(t, \omega)$  are   the  WVDs  of $f(t)$ and  $h_{\beta}(t)$   respectively defined as
	\begin{align}
		W_f(t, \omega)&=\int_{-\infty}^{+\infty}f(t+\frac{\tau}{2})f^*(t-\frac{\tau}{2})\mathrm{e}^{-j\omega \tau }\mathrm{d}\tau,\\
		W_{h_{\beta}}(t,\omega)&=\int_{-\infty}^{+\infty}h_{\beta}(t+\frac{\tau}{2})h_{\beta}^*(t-\frac{\tau}{2})\mathrm{e}^{-j \omega \tau }\mathrm{d}\tau.
	\end{align}	 
\end{lemma} 

The proof of Lemma 1 is similar  to the result in [16,18] for the spectrogram.

From Lemma 1, we can know that the  shape of the window,  as a TF filter or kernel,  is critical because it determines the energy distribution of the CT in multi-component non-stationary signal analysis.  A matching one with the TF  feature of an observed signal can help for obtaining a high-resolution TF representation [54,56,58,59]. 

To further analyze the effect of  chirp-based window on  CT, we consider a concrete window, i.e., the Gaussian window  $g(t)={(\sqrt{2\pi}\sigma)}^{-\frac{1}{2}}\mathrm{e}^{-\frac{t^2}{2\sigma^2}}$ ($\sigma>0$ is the variance of the window),  as this function leads to the optimal TF resolution and based on it the analytic expression of some TF transformations is easy to compute.  Denote the chirp-based  Gaussian  window   as 
\begin{equation}
	\begin{split}
		h_{\sigma,\beta}(t)=g(t)\mathrm{e}^{j\frac{\beta }{2}t^2}={(\sqrt{2\pi}\sigma)}^{-\frac{1}{2}}\mathrm{e}^{-\frac{t^2}{2\sigma^2}}\mathrm{e}^{j\frac{\beta }{2}t^2},
	\end{split}
\end{equation}
then  the corresponding CT can be  expressed by
\begin{equation}
	\begin{split}
		&C_{f}^{h_{\sigma, \beta}}(t, \omega)=\\&(\sqrt{2\pi}\sigma)^{-\frac{1}{2}}\int_{-\infty}^{+\infty}f(\mu)\mathrm{e}^{-\frac{(t-\mu)^2}{2\sigma^2}}\mathrm{e}^{-j\frac{\beta}{2}(\mu- t)^2}\mathrm{e}^{-j\omega(\mu- t)}\mathrm{d}\mu.
	\end{split}
\end{equation}
The WVD of this  window $h_{\sigma, \beta}(t)$ is  calculated as 
\begin{equation}
	\begin{split}
		W_{h_{\sigma, \beta}}(t,\omega)&=\int_{-\infty}^{+\infty}h_{\sigma, \beta}(t+\frac{\tau}{2})h_{\sigma, \beta}^*(t-\frac{\tau}{2})\mathrm{e}^{-j\omega \tau }\mathrm{d}\tau\\
		&=\sqrt{2}\mathrm{e}^{-\frac{t^2}{\sigma^2}-\sigma^2(\omega -\beta t)^2}.
	\end{split}
\end{equation}
Obviously, the shape of $W_{h_{\sigma, \beta}}(t,\omega)$ is  an oblique ellipse centered at the origin in the TF plane (see Fig. 1),  and  affected by two parameters, i.e., $\beta$ and $\sigma$. Maybe, the parameter $\beta$ determines   the direction of the  ellipse axis, and the parameter $\sigma$  determines the size  of the  ellipse.   To  make  this  issue  clear, let us consider a  level curve of $W_{h_{\sigma, \beta}}(t,\omega)$, which is given by
\begin{equation}
	\begin{split}
		\frac{t^2}{\sigma^2}+\sigma^2(\omega -\beta t)^2=C,
	\end{split}
\end{equation}
where $C$ is a   positive constant (i.e., percentage
of the  maximum energy). 

\begin{figure}[!tbph]
		\vspace{-0.2cm}
	\setlength{\belowcaptionskip}{-0cm}
	\centering
	\begin{minipage}{0.9\linewidth}
		\centerline{\includegraphics[width=1\textwidth]{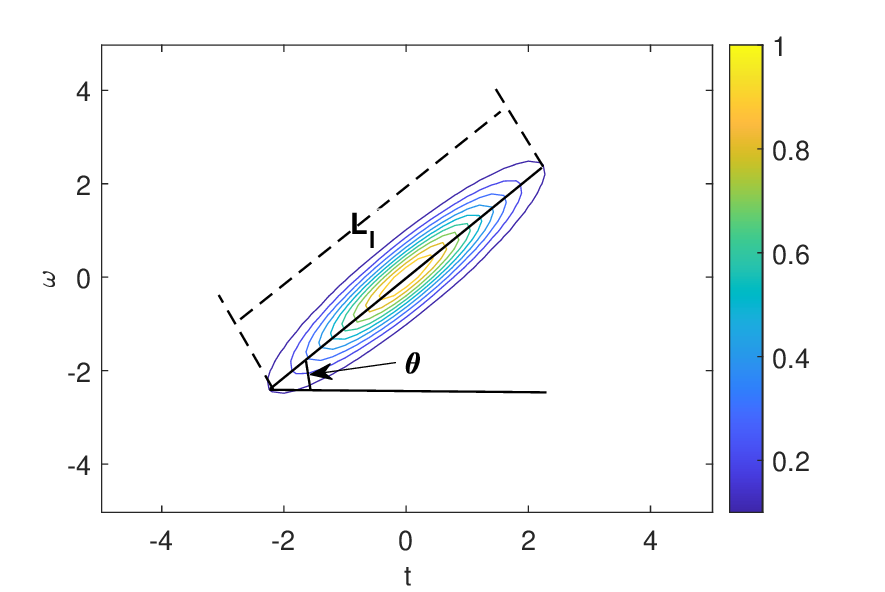}}
	\end{minipage}
	\caption{ \small Plot of $W_{h_{\sigma,\beta}}(t,\omega)$ when  $\sigma=1.5$ and $\beta=1$.}
\end{figure}

Clearly,   the shape of (9) becomes a circle when $\beta=0$ and $\sigma=1$. Except this  special case, the following  theorem  explains the influence of the parameters   $\sigma$  and  $\beta$ on the shape  of  ellipse (9).
\begin{theorem}
	For  ellipse (9), the length  of the  long axis is that  
	\begin{align}
		L_l=C({A\sigma^2}+\frac{2}{\sigma^2}+{\sqrt{A^2\sigma^4+4\beta^2}}), 
	\end{align}	 
	where $A=1+\beta^2-\frac{1}{\sigma^4}$,  and the angle $\theta$ ($\theta \in (-\frac{\pi}{2}, \frac{\pi}{2}]$)  between $t$-axis and the  long axis  (see Fig. 1)  satisfies
	\begin{align}
		& \lim_{\sigma\rightarrow +\infty} \tan \theta =\beta,\\
		& |\tan \theta| \geq \sqrt{\frac{1}{\sigma^4}-1}~  ~ \text{for}~ ~ 0<\sigma<1.	
	\end{align}	 
\end{theorem}

\begin{proof}
	See the  Supplementary Materials-\uppercase\expandafter{\romannumeral1}.
\end{proof}
 Theorem 1 throws up  some interesting results. On the one hand,  the length  $L_l$  increases with the  value of $\beta$  increasing,  and also  exhibiting  an increase with the decrease in  value of $\sigma$,  which may lead the CT to generate a bad TF result when dealing with strong modulation signals  because a large $\beta$ and small $\sigma$  are  required to match  its fast-varying  feature. On the other hand,   from the proof (please refer to equation (49) in the Supplementary Materials),   we can know that there does not exist a correspondence relationship that $	\tan \theta=\beta$ between the rotation angle $\theta$ and parameter  $\beta$. But  a large value of  $\sigma$ is  more advisable for CT because based on it  the major axis of $W_{h_{\sigma, \beta}}(t,\omega)$  can rotate  along any direction to approach that $	\tan \theta=\beta$. Moreover, relation (12) shows that $|\tan\theta|$ is bounded and it can not be taken as a sufficiently small value  when the input variance  is small (e.g., $\sigma\leq 0.1$), which limits the slope  range of  $W_{h_{\sigma, \beta}}(t,\omega)$, leading to  inaccurate  match in dealing with  short  signal segments containing smaller CR values.

\begin{figure}[!tbh]
			\vspace{-0.1 cm}
	\setlength{\belowcaptionskip}{-0.1 cm}
	\centering
	\begin{minipage}{0.48\linewidth}
		\centerline{\includegraphics[width=1\textwidth]{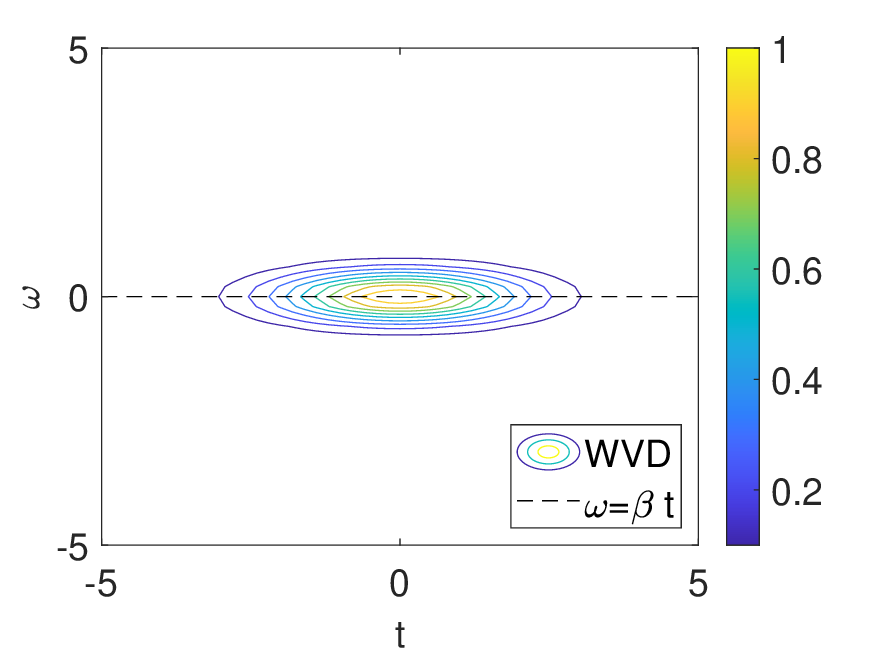}}
		\centerline{(a)}
	\end{minipage}
	\begin{minipage}{0.48\linewidth}
		\centerline{\includegraphics[width=1\textwidth]{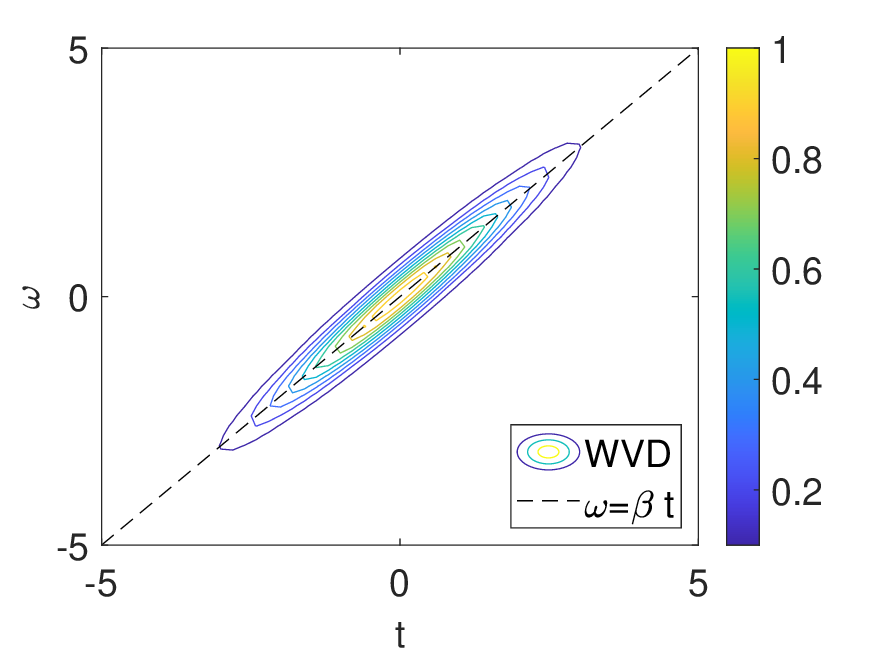}}
		\centerline{(b)}
	\end{minipage}
	\begin{minipage}{0.48\linewidth}
		\centerline{\includegraphics[width=1\textwidth]{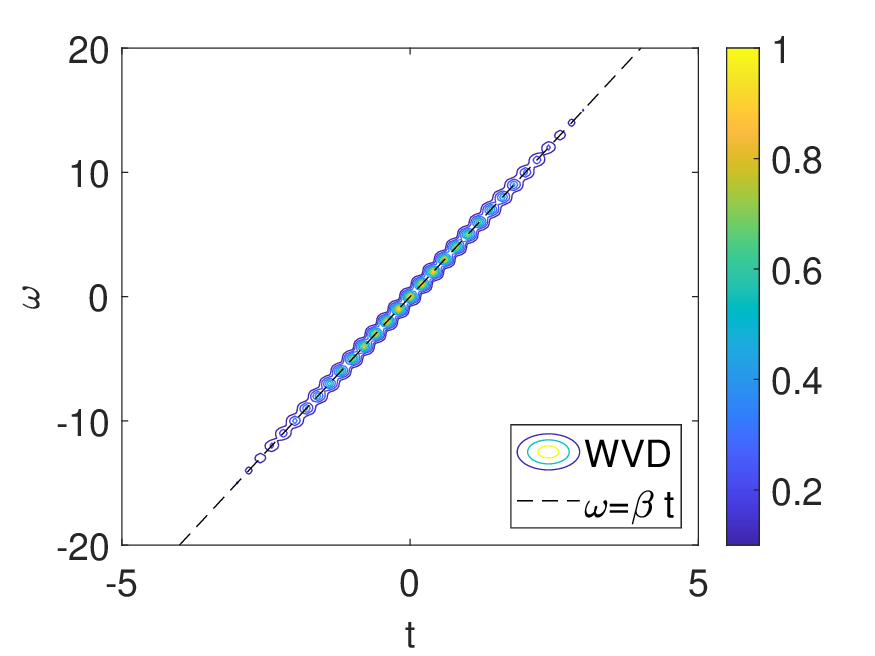}}
		\centerline{(c)}
	\end{minipage}
	\begin{minipage}{0.48\linewidth}
		\centerline{\includegraphics[width=1\textwidth]{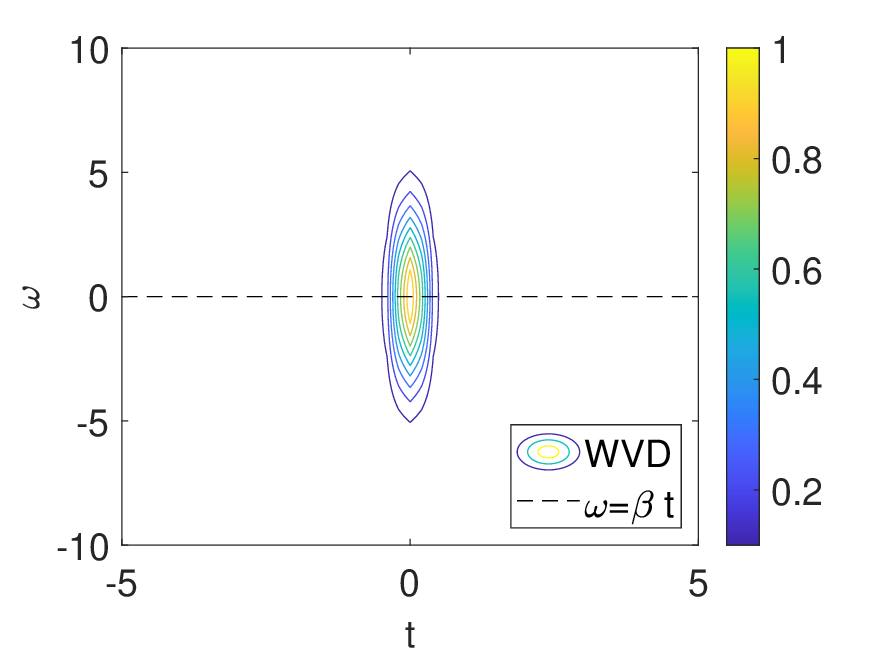}}
		\centerline{(d)}
	\end{minipage}
	\begin{minipage}{0.48\linewidth}
		\centerline{\includegraphics[width=1\textwidth]{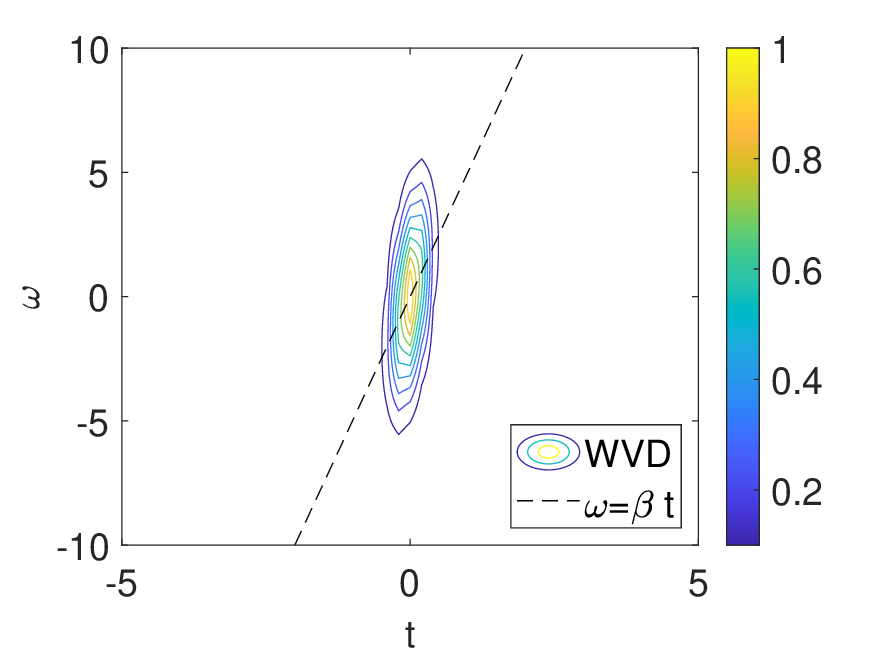}}
		\centerline{(e)}
	\end{minipage}
	\begin{minipage}{0.48\linewidth}
		\centerline{\includegraphics[width=1\textwidth]{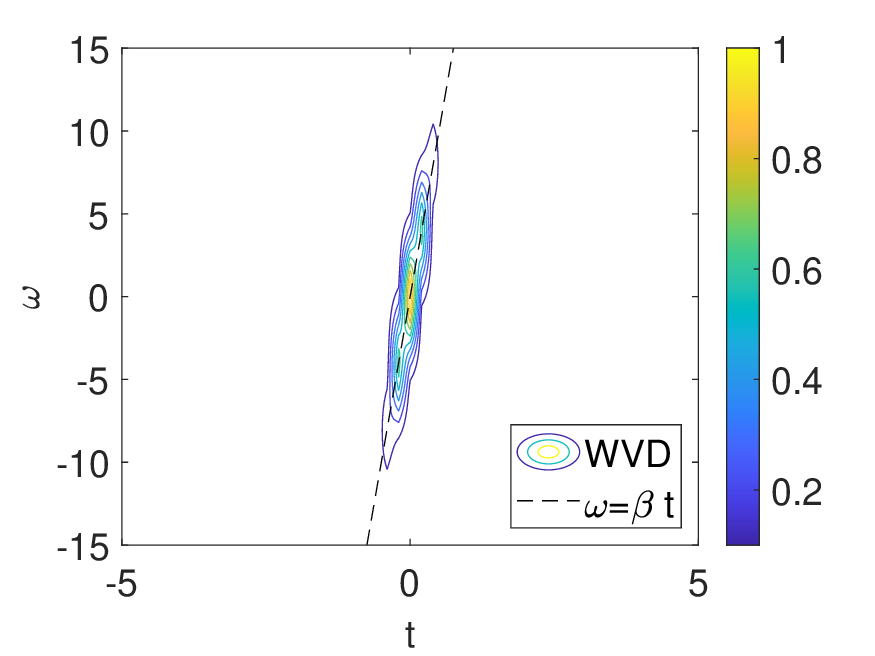}}
		\centerline{(f)}
	\end{minipage}
	\caption{ \small The WVDs of  $h_{\sigma, \beta}(t)$  using  various values of $\sigma$ and $\beta$. (a) The WVD with $\sigma=2$ and $\beta=0$, (b) the WVD with  $\sigma=2$ and $\beta=1$, (c) the WVD with   $\sigma=2$ and $\beta=5$, (d) the WVD with  $\sigma=0.3$ and $\beta=0$,  (e) the WVD with   $\sigma=0.3$ and $\beta=5$, (f) the WVD with    $\sigma=0.3$ and $\beta=20$.}
\end{figure}

Fig. 2 shows the WVDs of the chirp-based  Gaussian  window  using   various values of $\beta$ and $\sigma$.  The first row in Fig. 2  is  the results  corresponding  to a large   variance (i.e., $\sigma=2$)  over different values of $\beta$  ($\beta=0$, $\beta=1$, $\beta=5$). It can be seen from the results that the length of the major axis increases  with the increasing of  $\beta$, and the rotation angle $\theta$  approximately satisfies  the relationship that $ \tan\theta =  \beta$.  While  using  a  small 
$\sigma$ (e.g., $\sigma=0.3$),  the length of the major axis also increases  with the increasing of  $\beta$,  but the range of  rotation angle $\theta$  is limited in  a neighborhood of $\frac{\pi}{2}$ (see Fig. 2(d-f)).

To solve the existing problem of the standard CT, we will  discuss  the performance of a rotation-window CT in the following setion.

\subsection{CT with rotation   window}

In order to  obtain a more canonical TF filter, i.e., parameters $\sigma$ and $\beta$  have a clear role in determining the shape of $W_{h_{\sigma, \beta}}(t,\omega)$, we use  the fractional Fourier transform (FrFT)  [67,68]  to  rotate Gaussian window $g(t)$  with parameter $\beta$, which is given by   
\begin{equation}
	\begin{split}
		\tilde{h}_{\sigma,\beta}(t)=\left( \sqrt{2\pi}\sqrt{\frac{1+\sigma^2\beta^2}{\sigma(1+\beta^2)}}\right) ^{-\frac{1}{2}}\mathrm{e}^{-\frac{\sigma(1+\beta^2)}{2(1+\sigma^2\beta^2)}t^2}\mathrm{e}^{j\frac{\beta(1-\sigma^2) }{2(1+\sigma^2\beta^2)}t^2}.
	\end{split}
\end{equation}
Using this rotation window, the corresponding CT, named rotation-window   CT,  is that  
\begin{equation}
	\begin{split}
		{C}_{f}^{\tilde{h}_{\sigma,\beta}}(t, \omega)=\int_{-\infty}^{+\infty}f(\mu)	\tilde{h}_{\sigma,\beta}(t-\mu)\mathrm{e}^{-j\omega(\mu- t)}\mathrm{d}\mu.
	\end{split}
\end{equation}
Based on Lemma 1,  the squared modulus of  ${C}_{f}^{\tilde{h}_{\sigma,\beta}}(t, \omega)$ can be  expressed as
\begin{equation}
	|{C}_{f}^{\tilde{h}_{\sigma,\beta}}(t, \omega)|^2=\frac{1}{2\pi}W_f\ast\ast W_{\tilde{h}_{\sigma,\beta}}(t, \omega),
\end{equation}
where $W_{\tilde{h}_{\sigma,\beta}}$ is the WVD of $\tilde{h}_{\sigma,\beta}(t)$ as follows:
\begin{equation}
	\begin{split}
		W_{\tilde{h}_{\sigma,\beta}}(t,\omega)=\sqrt{2}\mathrm{e}^{-\frac{t^2}{\hat{\sigma}^2}-\hat{\sigma}^2(\omega -\hat{\beta} t)^2},
	\end{split}
\end{equation}
with  $\hat{\sigma}^2=\frac{1+\sigma^2\beta^2}{\sigma(1+\beta^2)}$, $\hat{\beta}=\frac{\beta(1-\sigma^2)}{1+\sigma^2\beta^2}$.

To discuss parameters $\sigma$ and  $\beta$  are how they affect  the shape of $W_{\tilde{h}_{\sigma,\beta}}(t,\omega)$, we also consider a level curve of $W_{\tilde{h}_{\sigma,\beta}}$,  which is
\begin{equation}
	\begin{split} 
		\frac{t^2}{\hat{\sigma}^2}+\hat{\sigma}^2(\omega -\hat{\beta} t)^2=C, 
	\end{split}
\end{equation}
where $C$ is a  positive constant. The follwing theorem shows the determination of the shape of  ellipse (17) by parameters $\sigma$ and  $\beta$.

\begin{theorem}
	For  ellipse (17), the length $L_l$ of the  long axis is   that 
	\begin{align}
		L_l=\begin{cases}
			\frac{2C}{\sigma},  & \text{if} ~ \sigma<1,\\
			2\sigma C,    &  \text{if} ~  \sigma\geq 1. 
		\end{cases}
	\end{align}	 
	The angle $\theta$ ($\theta \in (-\frac{\pi}{2}, \frac{\pi}{2}]$) between $t$-axis and the  long axis  satisfies
	\begin{align}
		\tan \theta =\begin{cases}
			\beta, & \text{if} ~~ 0<\sigma<1,\\
			-\frac{1}{\beta},  & \text{if} ~ \sigma>1, \beta\neq 0,\\
				+\infty,  & \text{if} ~ \sigma>1, \beta = 0.		
		\end{cases}
	\end{align}	 
\end{theorem}
\begin{proof}
	See the  Supplementary Materials-\uppercase\expandafter{\romannumeral2}.
\end{proof}

For the rotation window $\tilde{h}_{\sigma,\beta}(t)$, Theorem 2 proves that  the parameter $\sigma$  determines the size  of the  WVD of $\tilde{h}_{\sigma,\beta}(t)$, and the parameter $\beta$ determines  its  inclined direction. This property makes the role of $\beta$ and $\sigma$ more clear and is helpful for parameter setting.  Further explanation is given from numerical experiments,  which are  presented in  Fig. 3.  From the results,  we can see that the size of  the  WVDs  is unchangeable  when  fixing the value of $\sigma$.  Moreover, when using different values of $\beta$,  the rotation  of  the WVD  of  $\tilde{h}_{\sigma,\beta}(t)$  satisfies the relation (19) given in Theorem 2. 	
\begin{figure}[!tbh]
	\vspace{-0.1 cm}
	\setlength{\belowcaptionskip}{-0.1 cm}
	\centering
	\begin{minipage}{0.48\linewidth}
		\centerline{\includegraphics[width=1\textwidth]{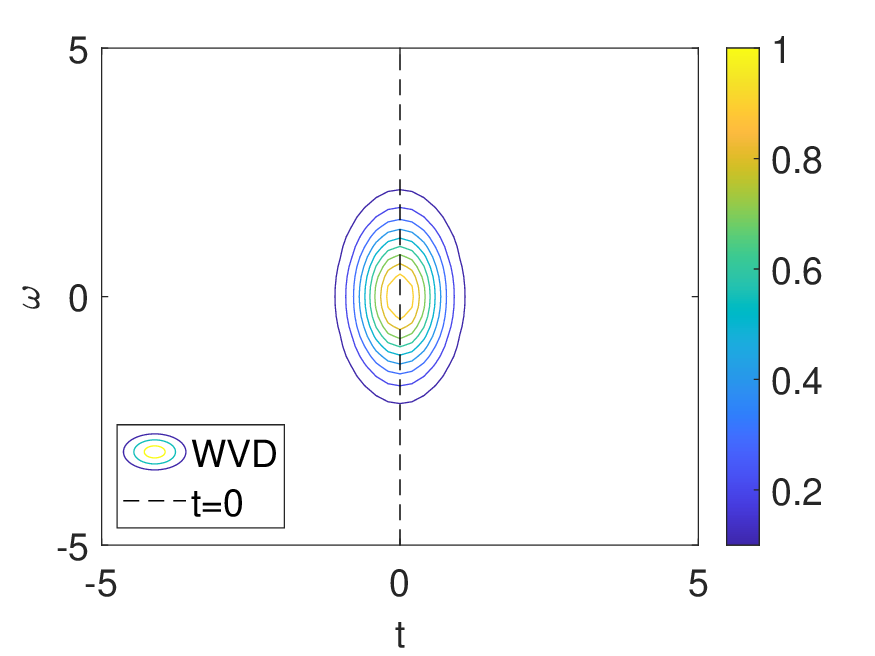}}
		\centerline{(a)}
	\end{minipage}
	\begin{minipage}{0.48\linewidth}
		\centerline{\includegraphics[width=1\textwidth]{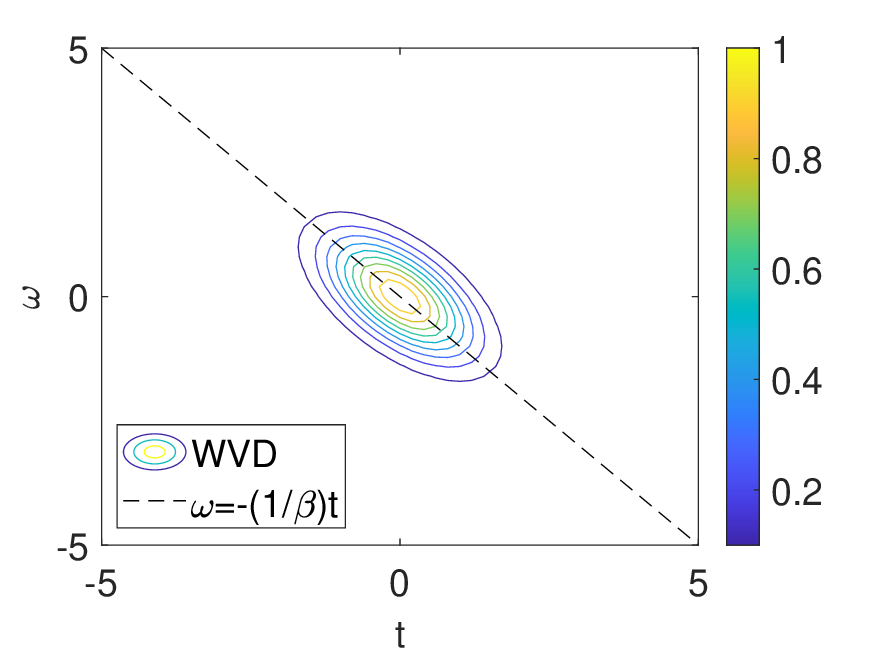}}
		\centerline{(b)}
	\end{minipage}
	\begin{minipage}{0.48\linewidth}
		\centerline{\includegraphics[width=1\textwidth]{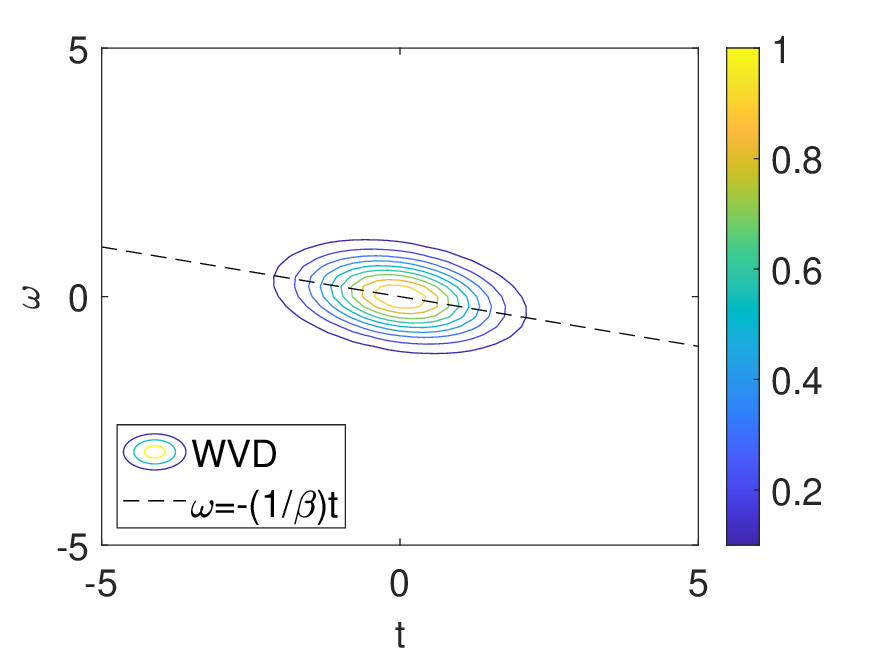}}
		\centerline{(c)}
	\end{minipage}
	\begin{minipage}{0.48\linewidth}
		\centerline{\includegraphics[width=1\textwidth]{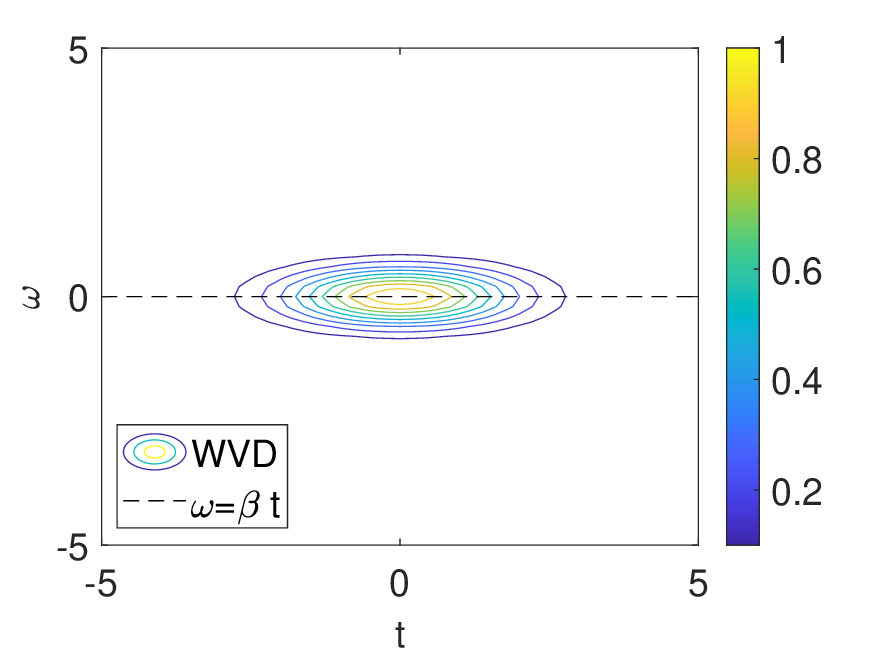}}
		\centerline{(d)}
	\end{minipage}
	\begin{minipage}{0.48\linewidth}
		\centerline{\includegraphics[width=1\textwidth]{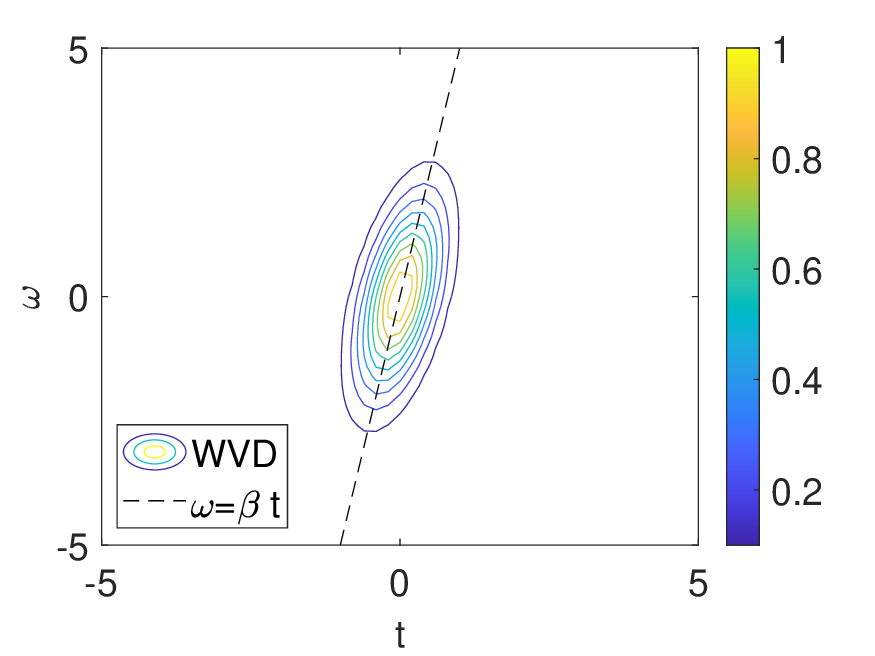}}
		\centerline{(e)}
	\end{minipage}
	\begin{minipage}{0.48\linewidth}
		\centerline{\includegraphics[width=1\textwidth]{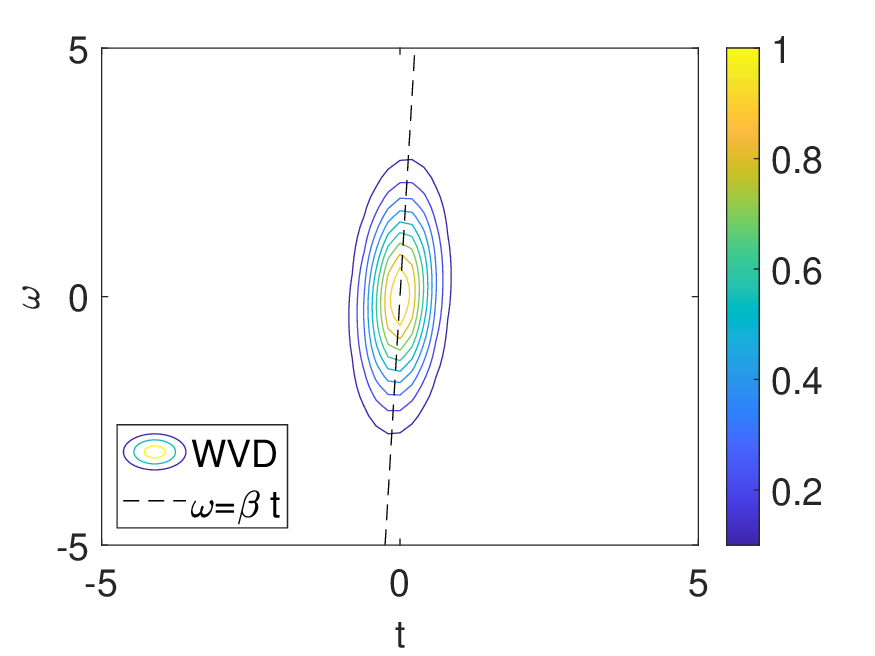}}
		\centerline{(f)}
	\end{minipage}
	\caption{ \small The WVDs of  $\tilde{h}_{\sigma,\beta}(t)$  using  various values of $\sigma$ and $\beta$. (a) The WVD with $\sigma=2$ and $\beta=0$, (b) the WVD with  $\sigma=2$ and $\beta=1$, (c) the WVD with   $\sigma=2$ and $\beta=5$, (d) the WVD with  $\sigma=0.3$ and $\beta=0$,  (e) the WVD with   $\sigma=0.3$ and $\beta=5$, (f) the WVD with    $\sigma=0.3$ and $\beta=20$.}
\end{figure}
\begin{figure*}[!tbh]
		\vspace{-0.cm}
	\setlength{\belowcaptionskip}{-0.cm}
	\centering
	\begin{minipage}{0.7\linewidth}
		\centerline{\includegraphics[width=1\textwidth]{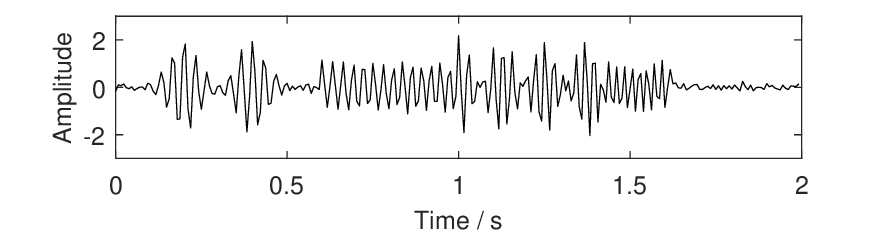}}
		\centerline{(a)}
	\end{minipage}
	\begin{minipage}{0.322\linewidth}
		\centerline{\includegraphics[width=1\textwidth]{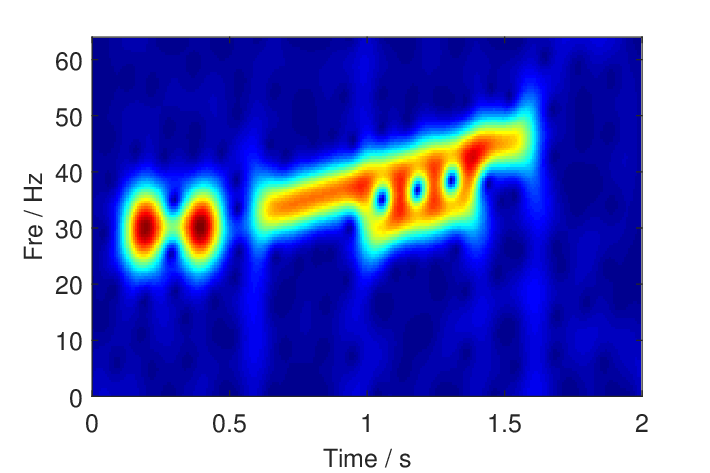}}
		\centerline{(b)}
	\end{minipage}
	\begin{minipage}{0.322\linewidth}
		\centerline{\includegraphics[width=1\textwidth]{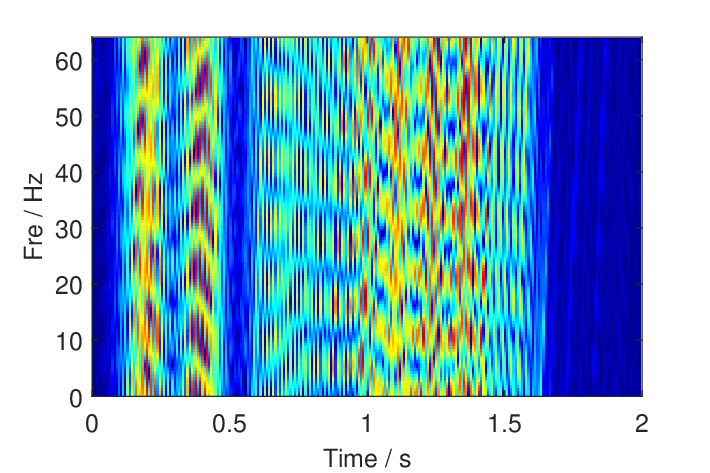}}
		\centerline{(c)}
	\end{minipage}
	\begin{minipage}{0.322\linewidth}
		\centerline{\includegraphics[width=1\textwidth]{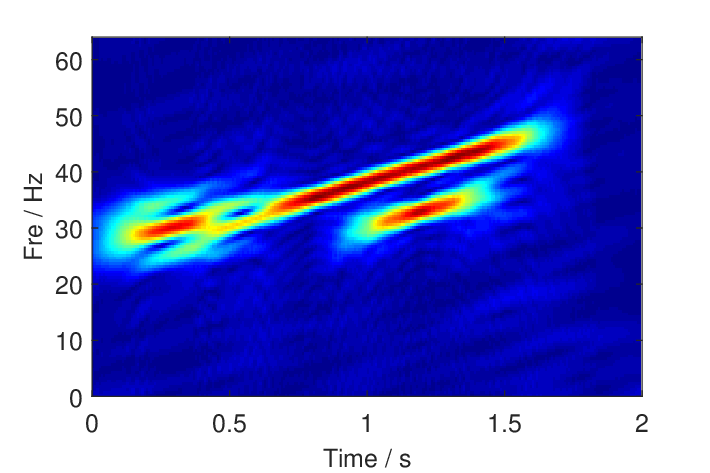}}
		\centerline{(d)}
	\end{minipage}
	\begin{minipage}{0.322\linewidth}
		\centerline{\includegraphics[width=1\textwidth]{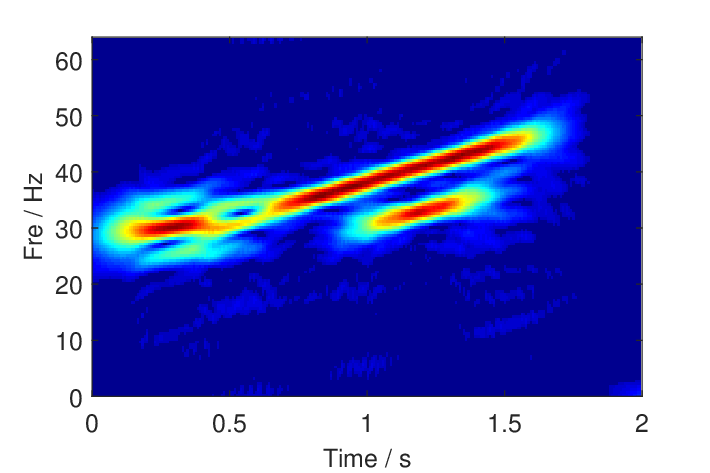}}
		\centerline{(e)}
	\end{minipage}
	\begin{minipage}{0.322\linewidth}
		\centerline{\includegraphics[width=1\textwidth]{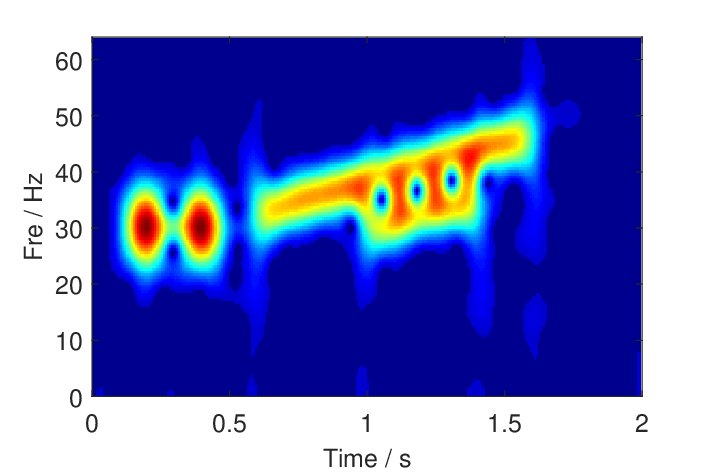}}
		\centerline{(f)}
	\end{minipage}
	\begin{minipage}{0.322\linewidth}
		\centerline{\includegraphics[width=1\textwidth]{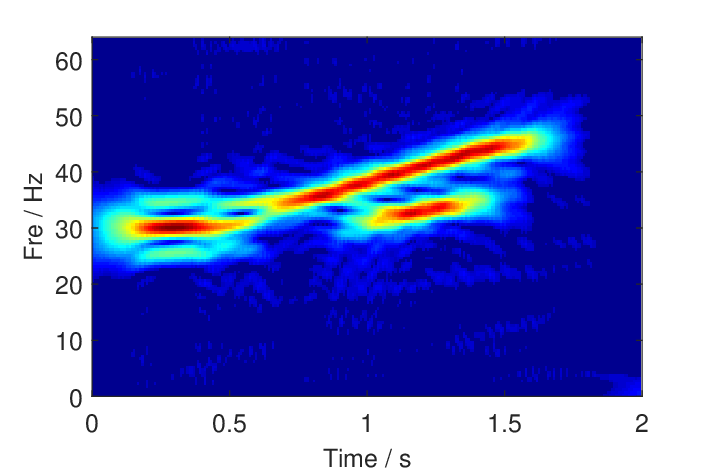}}
		\centerline{(g)}
	\end{minipage}
	\caption{ \small Evaluation of TF resolution on a known signal structure. (a) Waveform of the test signal, (b) CT with  $\sigma=0.1$ and $\beta=7$, (c) CT with $\sigma=0.1$ and $\beta=1000$, (d) CT with  $\sigma=0.3$ and $\beta=7$, (e) the rotation-window CT with $\sigma=0.1$ and $\beta=7$, (f) the rotation-window CT with $\sigma=0.1$ and $\beta=1000$, (g) the rotation-window CT with $\sigma=0.3$ and $\beta=7$.}
\end{figure*}

\subsection{Comparison}

The section above mainly discusses the effect of the chirp-based  Gaussian  window and the  rotating Gaussian window  on  the CT from theory, which provides important theoretical guidance for CT to choose suitable parameters so as to  obtain a high-resolution TF representation.  In the following, we test the performance of CT and rotation-window CT by  an example. The waveform of the test signal is illustrated in  Fig. 4(a), which contains two chirp signals (CR$=$7) and two impulse  signals. We compute the TF representations of this test  signal using CT and rotation-window CT with different parameters (see Fig. 4(b-g)). From the results, we can see that the  CT can fully separate  the two  impulse signals when employing a smaller value of $\sigma$, but a small one  results in  a poor frequency resolution such that the TF  mixing of chirp  components occurs (see Fig. 4(b)).  A relatively large $\sigma$ following with  good  CR estimation can resolve the two chirps, and yet obscures   the TF information of the two impulses (see Fig. 4(d)).   Moreover, using  a large $\beta$ to match the fast-varying feature of impulse  signals degrades the TF performance of the CT, as presented in Fig. 4(c).  By contrast, it follows that the  rotation-window CT provides  a better  time resolution when  using a small  $\sigma$ and a  large $\beta$ (see Fig. 4(f)), and can separate   the two  chirps   even  with a small  $\sigma$ (see Fig. 4(e)).   In fact, an advantage of the rotation-window  CT  over the standard CT is that the  parameters in  the  window  become relative, forming  effective  combinations  such that the window shape can  rotate  along any direction  in the TF plane only by adjusting  the  parameter  $\beta$.

Nevertheless, the  rotation-window CT  is actually the standard  CT  with a special parameter combination (i.e., $\hat{\sigma}$ and  $\hat{\beta}$),  thus it is still  limited by  the Heisenberg uncertainty principle and can not localize precisely a signal in both time and frequency. As presented in Fig. 4,  the window of these two methods  matches the pulse (chirp) components well but smears the chirp (pulse) components.   This illustrates the fundamental tradeoff of the CT: it  is difficult  to get both good time and good frequency resolution using a single fixed window.

\section{Multi-resolution  chirplet transform}
To overcome  the TF resolution problem of the CT, more effective possibilities  should be considered.    
The following section  utilizes a combination of multiple CTs with different parameter combinations to improve the TF resolution of the CT, as it can localize the signal in both time and frequency better than it is possible with any single CT.  Moreover,  the combined one  can  reduce the computational complexity of adaptive CTs in which the window function is adjusted at every  TF  location. 

\subsection{Multi-resolution chirplet transform}
Let us  consider a two-component non-stationary signal 
including  an impulse signal and a chirp signal, which is given by 
\begin{equation}
	\begin{split}
		&f(t)=f_1(t)+f_2(t),\\&f_1(t)=A_1\delta(t-t_0), ~f_2(t)=A_2(t) \mathrm{e}^{j(at+\frac{b}{2}t^2)},
	\end{split}
\end{equation}
where  $A_1>0$, $A_2(t)\geq 0$, $a, b\in\mathbb{R}$,  and $\delta(\cdot)$ is the  Dirac delta function.

Assuming  that around  $\mu=t$  the amplitude  $A_2(t)$ is  slowly varying and can be  well approximated  by $A_2(t+\mu)\approx A_2(t)$, then the CT of signal (20) under the Gaussian window $g(t)={(\sqrt{2\pi}\sigma)}^{-1}\mathrm{e}^{-\frac{t^2}{2\sigma^2}}$ is computed as
\begin{align} 
	&C_{f}^{h_{\sigma,\beta}}(t,\omega)=\int_{-\infty}^{+\infty}f(\mu)g^*(t-\mu)\mathrm{e}^{-j\frac{\beta}{2}(\mu- t)^2}\mathrm{e}^{-j\omega(\mu- t)}\mathrm{d}\mu \nonumber\\
	&={(\sqrt{2\pi}\sigma)}^{-1}\int_{-\infty}^{+\infty}f_1(t+\mu)\mathrm{e}^{-\frac{\mu^2}{2\sigma^2}}\mathrm{e}^{-j\frac{\beta}{2}\mu^2}\mathrm{e}^{-j\omega\mu}\mathrm{d}\mu\nonumber\\&+{(\sqrt{2\pi}\sigma)}^{-1}\int_{-\infty}^{+\infty}f_2(t+\mu)\mathrm{e}^{-\frac{\mu^2}{2\sigma^2}}\mathrm{e}^{-j\frac{\beta}{2}\mu^2}\mathrm{e}^{-j\omega\mu}\mathrm{d}\mu\nonumber\\
	&={(\sqrt{2\pi}\sigma)}^{-1}(A_1\mathrm{e}^{-\frac{(t_0-t)^2}{2\sigma^2}}\mathrm{e}^{-j\frac{\beta}{2}(t_0-t)^2}\mathrm{e}^{-j\omega(t_0-t)}\nonumber+\\&\mathrm{e}^{j(at+\frac{b}{2}t^2)}\int_{-\infty}^{+\infty}A_2(t+\mu)\mathrm{e}^{-\frac{\mu^2}{2\sigma^2}}\mathrm{e}^{-j\frac{(\beta-b)}{2}\mu^2}\mathrm{e}^{-j(\omega-a-bt)\mu}\mathrm{d}\mu)\nonumber\\
	&	\approx {(\sqrt{2\pi}\sigma)}^{-1}A_1\mathrm{e}^{-\frac{(t_0-t)^2}{2\sigma^2}}\mathrm{e}^{-j\frac{\beta}{2}(t_0-t)^2}\mathrm{e}^{-j\omega(t_0-t)}\nonumber\\& +{(\sqrt{2\pi}\sigma)}^{-1}f_2(t)\int_{-\infty}^{+\infty}\mathrm{e}^{-\frac{\mu^2}{2\sigma^2}}\mathrm{e}^{-j\frac{(\beta-b)}{2}\mu^2}\mathrm{e}^{-j(\omega-a-bt)\mu}\mathrm{d}\mu\nonumber\\
	&\xlongequal[]{\beta=b}  {(\sqrt{2\pi}\sigma)}^{-1} A_1\mathrm{e}^{-\frac{(t_0-t)^2}{2\sigma^2}}\mathrm{e}^{-j\frac{b}{2}(t_0-t)^2}\mathrm{e}^{-j\omega(t_0-t)}\nonumber\\& +f_2(t) \mathrm{e}^{-\frac{\sigma^2}{2}(\omega-a-bt)^2}.
\end{align} 
If  $f_1(t)$ and $f_2(t)$ are further separable, e.g.,  $A_2(t)\approx0$ when  $t\in [t_0-\sigma, t_0+\sigma]$,  then the modulus of  CT   can be expressed as 
\begin{align} 
	&	|C_{f}^{h_{\sigma,b}}(t, \omega)| \nonumber\\&\approx{(\sqrt{2\pi}\sigma)}^{-1}A_1\mathrm{e}^{-\frac{1}{2\sigma^2}(t_0-t)^2} +A_2(t) \mathrm{e}^{-\frac{\sigma^2}{2}(\omega-a-bt)^2}.
\end{align} 
From (22),  we can know that high time-resolution can be achieved for $f_1(t)$ with a small value of $\sigma$,  but that causes a degradation of frequency resolution. Conversely,  a large $\sigma$ increases  frequency resolution  for $f_2(t)$ but at the expense of temporal resolution of the CT of  $f_1(t)$ (see Fig. 4(b,d) for numerical illustration).   

To overcome this problem, let us consider  the product of two  CTs with  different parameters $\sigma_1$  and $\sigma_2$, which is calculated as 
\begin{align} 
	&|C_{f}^{h_{\sigma_1,b}}(t, \omega)\times C_{f}^{h_{\sigma_2,b}}(t, \omega)|\approx \frac{1}{2\pi \sigma_1\sigma_2}A_1^2\mathrm{e}^{-\frac{\sigma_1^2+\sigma_2^2}{2\sigma_1^2\sigma_2^2}(t_0-t)^2}\nonumber\\& + A_2(t)^2 \mathrm{e}^{-\frac{\sigma_1^2+\sigma_2^2}{2}(\omega-a-bt)^2}.
\end{align} 
Comparing (22) with (23), we find that the  product  has better time resolution for $f_1(t)$ and  increased frequency resolution for $f_2(t)$  (here we can  suppose  $\sigma$ and  $\sigma_1$ (or $\sigma$ and $\sigma_2$) take  the same value). To make this point more clear,   we give the combined CT of the test signal (see Fig. 4(a)) in Fig. 5(a).  Obviously, the  combined result  achieves a sharper  TF representation than that by  the single CT with $\sigma=0.1$ or $\sigma=0.3$, as  presented in Fig. 5(a) and Fig. 4(b,d).   
\begin{figure}[!tbph]
	\centering
	\begin{minipage}{0.493\linewidth}
		\centerline{\includegraphics[width=1\textwidth]{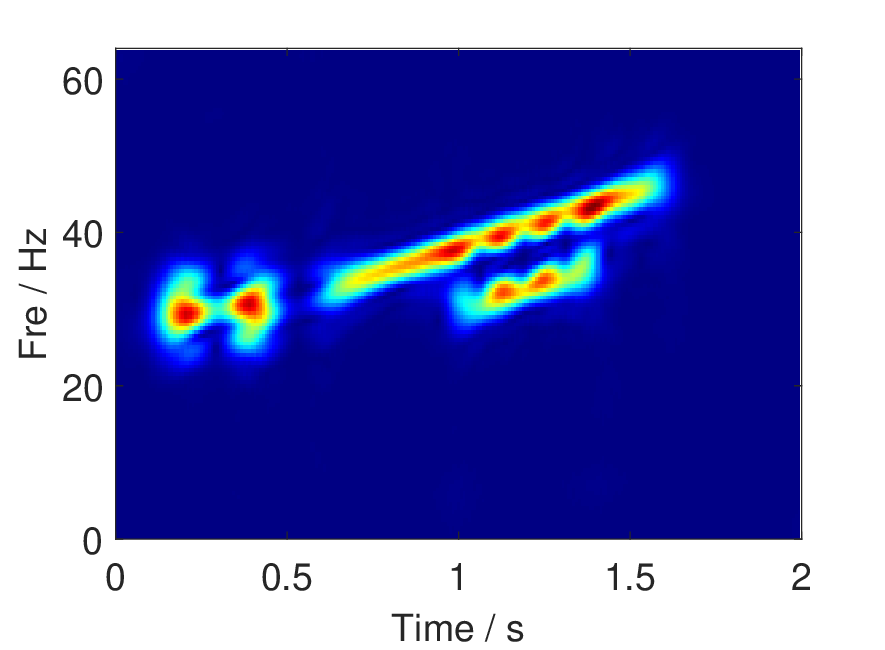}}
		\centerline{(a)}
	\end{minipage}
	\begin{minipage}{0.493\linewidth}
		\centerline{\includegraphics[width=1\textwidth]{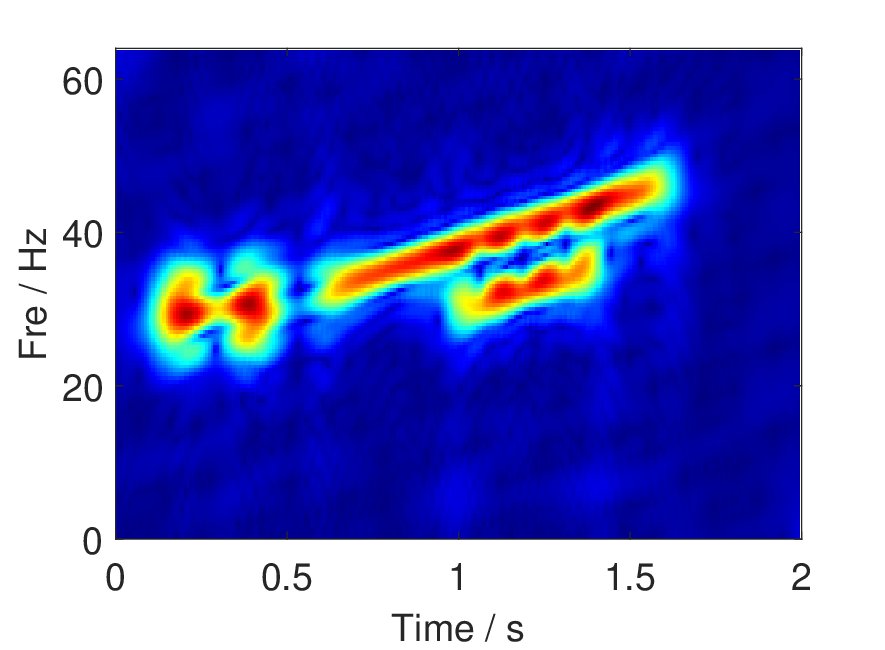}}
		\centerline{(b)}
	\end{minipage}
	\caption{ \small The results of the combined  CT  with $\sigma=0.1$ and $\sigma=0.3$ (here $\beta=7$). (a) The product of the two CTs, (b) the  geometric  mean of the two CTs. }
\end{figure}

Unlike the adaptive CTs [56,59],  the combined  CT  only uses two parameters $\sigma_1$  and $\sigma_2$ to yield a high-resolution result, thus it reduces the computational complexity effectively. Meanwhile, it is worth pointing out that,  when computing the product of multiple CTs to increase  TF resolution, one existing problem  is that the magnitude of the combined TF representation may become increasingly polarized with the increase of the number of  product-times, i.e., a small  TF energy becomes smaller and smaller,  and a large one is getting bigger and bigger.  Therefore it is undesirable to multiply too many  CTs for the combined  CT.  An effective way to solve the magnitude problem  is to compute its  geometric mean [5,50,51]. However, the mean step inevitably reduces  time and frequency resolution of the TF representation (see Fig. 5(b)). Nevertheless, we can derive the following conclusion based on (22) and (23),  which shows the superiority of the combined  CT over any single CT.
\begin{theorem}
For  signal  (20),  given   any $\sigma_1>0$,  $\sigma_2>0$, $\sigma_1\neq \sigma_2$, there does not exist $\sigma>0$ such that  $|C_{f}^{h_{\sigma,b}}(t, \omega)|$ has the same time  and frequency resolution as $|C_{f}^{h_{\sigma_1,b}}(t, \omega)\times C_{f}^{h_{\sigma_2,b}}(t, \omega)|^{\frac{1}{2}}$.
\end{theorem}

The proof of Theorem 3  is evident  based on (22) and (23).  If we assume  $|C_{f}^{h_{\sigma,b}}(t, \omega)|$ has the same  frequency resolution as $|C_{f}^{h_{\sigma_1,b}}(t, \omega)\times C_{f}^{h_{\sigma_2,b}}(t, \omega)|^{\frac{1}{2}}$ for $f_2(t)$, which means that $\sigma^2=\frac{\sigma_1^2+\sigma_2^2}{2}$. In this case,  the duration of  $|C_{f}^{h_{\sigma,b}}(t, \omega)|$ for $f_1(t)$ is $\sqrt{\frac{\sigma_1^2+\sigma_2^2}{2}}$. While the time spread   of $|C_{f}^{h_{\sigma_1,b}}(t, \omega)\times C_{f}^{h_{\sigma_2,b}}(t, \omega)|^{\frac{1}{2}}$  for $f_1(t)$  is $\frac{\sqrt{2}\sigma_1\sigma_2}{\sqrt{\sigma_1^2+\sigma_2^2}}$ satisfying $\frac{\sqrt{2}\sigma_1\sigma_2}{\sqrt{\sigma_1^2+\sigma_2^2}}<\sqrt{\frac{\sigma_1^2+\sigma_2^2}{2}}$ for $\sigma_1\neq \sigma_2$, which means that $|C_{f}^{h_{\sigma_1,b}}(t, \omega)\times C_{f}^{h_{\sigma_2,b}}(t, \omega)|^{\frac{1}{2}}$  has a better    time  resolution    than  $|C_{f}^{h_{\sigma,b}}(t, \omega)|$  for $f_1(t)$  if  they have the same  frequency resolution  for $f_2(t)$. Similarly, we can prove that  $|C_{f}^{h_{\sigma_1,b}}(t, \omega)\times C_{f}^{h_{\sigma_2,b}}(t, \omega)|^{\frac{1}{2}}$  provides a better    frequency resolution  than  $|C_{f}^{h_{\sigma,b}}(t, \omega)|$   for $f_2(t)$   if  they have the same time resolution  for $f_1(t)$. Furthermore, when specifying $\sigma_2=2\sigma_1$,  the product of the time pread   of  $f_1(t)$   and  the  bandwidth of  $f_2(t)$   can  be computed approximately as $\frac{2}{5}$  for  $|C_{f}^{h_{\sigma_1,b}}(t, \omega)\times C_{f}^{h_{\sigma_2,b}}(t, \omega)|^{\frac{1}{2}}$, smaller than the value of  $|C_{f}^{h_{\sigma,b}}(t, \omega)|$ that corresponds to $\frac{1}{2}$.   This theorem  in fact shows the  combined multi-resolution CT  resolves  the joint TF density better than it is possible with a single CT representation in dealing with signal (20). The numerical experiment also verifies the theoretical results, one can see  Fig. 4(b,d)  and  Fig. 5.  We should also highlight the fact that the combined procedure does  not  break the uncertainty principle completely,  because there is still energy spread in the combined TF distribution.

In practical application,  most  signals  contain multiple components with time-varying information, so we   should  extend  the combined CT based on two CTs to multi-resolution  case. 

Define $S_{\sigma,\beta}^K$ as a set of $h_{\sigma,\beta}(t)$  with different $\sigma$ and $\beta$:
\begin{align} 
	S_{\sigma,\beta}^m=\left\lbrace h_{\sigma,\beta}(t)| (\sigma, \beta) \in \{(\sigma_1, \beta_1), (\sigma_2, \beta_2), \cdots, (\sigma_m, \beta_m)\} \right\rbrace,
\end{align} 
where $m$  is a positive integer denoting the number of chirp-based  Gaussian windows.  The multi-resolution chirplet transform (MrCT)  of  a signal $f(t)$  with regard to $S_{\sigma,\beta}^K$ is then defined as
\begin{align} 
	MrC_{f}(t, \omega)&=\left( \prod_{i=1}^mC_{f}^{h_{\sigma_i,\beta_i}}(t, \omega)\right) ^{\frac{1}{m}},
\end{align} 
MrCT is to  compute the geometric  mean  of $m$ CTs,  which corresponds to  the optimal solution in the Kullback-Leibler divergence [69] sense.
\begin{theorem}
	Given  $m$ CTs: $C_{f}^{h_{\sigma_i,\beta_i}}(t, \omega)$, $i=1,2,\cdots,m$,  the  $|MrC_{f}(t, \omega)|$ is the optimal solution of the following optimization problem:
	\begin{align} 
		\min_{P(t, \omega)\geq 0}  \sum_{i=1}^{m}\| P(t,\omega)- |C_{f}^{h_{\sigma_i,\beta_i}} (t,\omega)|\|_{GKL}, 
	\end{align} 
	where $\| P(t,\omega)- Q(t,\omega)\|_{GKL}= \int \int P(t,\omega) \ln \frac{P(t,\omega)}{Q(t,\omega)} -P(t,\omega)+Q(t,\omega) \mathrm{d}t \mathrm{d}\omega$ for $P(t, \omega)\geq 0$ and $Q(t, \omega)\geq 0$. 
\end{theorem}

\begin{proof}
The derivative with respect to $P(t, \omega)$  of optimization problem  (26) must equal zero at the solution point:  
\begin{align} 
	\sum_{i=1}^{m}\ln \frac{P(t,\omega)}{|C_{f}^{h_{\sigma_i,\beta_i}} (t,\omega) |}=0.
\end{align}        
Solving this equation for $P(t, \omega)$  yields the optimal  solution
$$P(t, \omega)=\left( \prod_{i=1}^m|C_{f}^{h_{\sigma_i,\beta_i}}(t, \omega)|\right) ^{\frac{1}{m}}=|MrC_{f}(t, \omega)|.$$ 
\end{proof} 

It is remarkable here that Theorem 4  provides a more profound interpretation  for  the   MrCT: it can be seen as  the ``optimal'' way to combine individual CTs.  There are, of course,   other  measures, such as $L_2$-norm and maximum cross-entropy [51],  used to combine  information from multiple CTs, which  are worthy of continued investigation.   

\subsection{Parameters selection}
Determine  the values of  parameters $(\sigma_i, \beta_i)$ ($i=1,2, \cdots$, $m$)  is crucial for MrCT. We discuss how to select  effective  values in this section.  

It is well-known that 
the optimal  values of  parameters $\sigma $ and $\beta $ are signal-dependent  such that the window can adapt to the signal being analyzed. Several estimation methods  of  signal  features (e.g., CR)  are  available in the literature [35,46,54,56,58], while which  generally focus on its continuously varying  characteristics.  
We propose that the MrCT is  computed  based on several discrete values $(\sigma_i, \beta_i)$  ($i=1, 2, \dots, m$), over a range of values that includes its maximum and minimum acceptable values.  
In order to get robust estimates that  contain the main information of the signal, we utilize the  chirp-Fourier transform (CFT) [70] to detect  the signal  CR  information, which is defined as
\begin{equation}\label{key1}
	CFT_{f}(\omega,\beta)=\int_{-\infty}^{+\infty}f(\mu)\mathrm{e}^{-j\frac{\beta}{2}\mu^2}\mathrm{e}^{-j\omega\mu}\mathrm{d}\mu.	
\end{equation}
When the chirp rate and the harmonic frequency is  matched to the signal,  the magnitude of CFT  appears a peak [70].    Therefore we  can obtain $m$  CR  estimations, which  basically correspond  to the varying tendency  of  the  frequency  of the signal,  by  detecting  the first $m$  peaks of $|CFT_{f}(\omega,\beta)|$,  denoted as  $(\omega^*_i, \beta^*_i)$ ($i=1,2,\cdots,m$).   Parameters $\beta_i$ ($i=1,2,\cdots,m$) of MrCT   are   then  taken  as
\begin{align} 
	\beta_i= \beta^*_i,   ~~i=1,2,\cdots,m. 
\end{align} 

The CR  estimations, in turn, can be utilized in selecting the appropriate $\sigma$. Cohen in [10] presented  the relationship between an optimal window width and the CR of a signal. For the Gaussian window case, a width is described as a standard deviation and can be defined as follows:
\begin{align} 
	\sigma(t)=\frac{1}{\sqrt{2\pi |\phi''(t)|}},
\end{align} 
where $\phi''(t)\neq 0$. Based on this fact, parameters $\sigma_i$ ($i=1,2,\cdots,m$) of MrCT are determined by 
\begin{align} 
	\sigma_i= \frac{C_\sigma}{\sqrt{2\pi |\beta^*_i|}}, ~~i=1,2,\cdots,m, 
\end{align} 
where $C_\sigma>0$ is used to adjust   the $\sigma$ in a more reasonable interval. Hence,  the values of  parameters $(\sigma_i, \beta_i)$ ($i=1,2,\cdots,m$)  in MrCT is determined  by 
\begin{align} 
	(\sigma_i, \beta_i)= (\frac{C_\sigma}{\sqrt{2\pi |\beta^*_i|}}, \beta^*_i), ~~i=1,2,\cdots,m.
\end{align} 

Alternatively,  the parameters $\sigma_i$ ($i=1,2,\cdots,m$) also  can be chosen multiplicatively  or additively, i.e., $\sigma_i=i \sigma_1$ or $\sigma_i= \sigma_1+i \Delta \sigma $,  as presented in [5], and further  matched with the detected  $\beta_i$ ($i=1,2,\cdots,m$)  depending on the relation (30), even though the equality does not hold at this time.  

Number  $m$ is generally  determined based on prior information or estimation from a coarse TF representation. This number is   tolerant and usually  much smaller than the length of the  analyzed signal,  which  is  one important reason  to  reduce the  computational complexity of the adaptive CTs.

\subsection{Computational complexity}
The MrCT is to compute the geometric  mean  of $m$ CTs,   so its computational complexity mainly comes from the CT calculations with different windows.  To be specific, we assume that a signal of  $N$ samples is utilized, then the  CTs with different $(\sigma_i, \beta_i)$ ($i=1,2,\cdots,m$) can be implemented by FFT and they require $O(mN^2\log_2N)$ operations. Therefore, the total computing complexity of MrCT is $O(mN^2\log_2N)$.

\section{Multi-resolution synchroextracting chirplet transform }

In this section,  we will develop a novel analysis  approach to enhance the TF readability of the proposed  MrCT.  

For the convenience of the following explanation, we consider the case of a
chirp signal $f(t) = A_0\mathrm{e}^{j\phi(t)}$, where $A_0>0$ is a constant, $\phi(t)=at+\frac{b}{2}t^2$. 
The CT of  this  chirp signal   under the Gaussian window $g(t)= {(\sqrt{2\pi}\sigma)}^{-\frac{1}{2}}\mathrm{e}^{-\frac{t^2}{2\sigma^2}}$   is given  by [30]
\begin{equation}
	\begin{split}
		&C_{f}^{h_{\sigma, \beta}}(t, \omega)\\&={(\sqrt{2\pi}\sigma)}^{-\frac{1}{2}}\int_{-\infty}^{+\infty}f(\mu)\mathrm{e}^{-\frac{(t-\mu)^2}{2\sigma^2}}\mathrm{e}^{-j\frac{\beta}{2}(\mu- t)^2}\mathrm{e}^{-j\omega(\mu- t)}\mathrm{d}\mu\\
		&=f(t)\sqrt{\frac{\sqrt{2\pi}\sigma}{1+j\sigma^2(\beta-b)}}\exp\left(-\frac{\sigma^2(\omega-\phi'(t))^2}{2(1+j\sigma^2(\beta-b))}\right),
	\end{split}
\end{equation}
where $\phi'(t)=a+bt$.
From (33), we can obtain that
\begin{equation}
	\begin{split}
		\frac{\partial}{\partial t}C_{f}^{h_{\sigma, \beta}}(t, \omega)=C_{f}^{h_{\sigma, \beta}}(t, \omega)\left( j \phi'(t) +\frac{b\sigma^2(\omega-\phi'(t))}{1+j\sigma^2(\beta-b)} \right).
	\end{split}
\end{equation}
Due to that 
\begin{align} 
		&\frac{\partial}{\partial t}C_{f}^{h_{\sigma, \beta}}(t, \omega)\nonumber \\&=\frac{\partial}{\partial t}\left( \frac{1}{\sqrt{\sqrt{2\pi}\sigma}}\int_{-\infty}^{+\infty}f(\mu)\mathrm{e}^{-\frac{(t-\mu)^2}{2\sigma^2}}\mathrm{e}^{-j\frac{\beta}{2}(\mu- t)^2}\mathrm{e}^{-j\omega(\mu- t)}\mathrm{d}\mu\right)\nonumber\\& =(\frac{1}{\sigma^2}+j\beta)C_{f}^{th_{\sigma, \beta}}(t, \omega)+j\omega C_{f}^{h_{\sigma, \beta}}(t, \omega),
\end{align}
where $C_{f}^{th_{\sigma, \beta}}(t, \omega)={(\sqrt{2\pi}\sigma)}^{-\frac{1}{2}}\int_{-\infty}^{+\infty}f(\mu)(\mu-t)\mathrm{e}^{-\frac{(\mu-t)^2}{2\sigma^2}}$ $\mathrm{e}^{-j\frac{\beta}{2}(\mu- t)^2}\mathrm{e}^{-j\omega(\mu- t)}\mathrm{d}\mu$.  Combining (34) and (35),  we can obtain 
\begin{equation}
	\begin{split}
		&(\frac{1}{\sigma^2}+j\beta)C_{f}^{th_{\sigma, \beta}}(t, \omega)\\=&C_{f}^{h_{\sigma, \beta}}(t, \omega)\left( j (\phi'(t)-\omega) +\frac{b\sigma^2(\omega-\phi'(t))}{1+j\sigma^2(\beta-b)} \right).
	\end{split}
\end{equation}
For different combinations $(\sigma_i, \beta_i)$ ($i=1,2,\cdots,m$),  we can derive  the following expression from  equation (36) as
\begin{equation}
	\begin{split}
		&\frac{	(\prod_{i=1}^m C_{f}^{th_{\sigma_i, \beta_i}}(t, \omega))^{\frac{1}{m}}}{	MrC_{f}(t, \omega)}\\=&\left( \prod_{i=1}^m(\frac{1}{\sigma_i^2}+j\beta_i)^{-1}\left( j (\phi'(t)-\omega) +\frac{b\sigma_i^2(\omega-\phi'(t))}{1+j\sigma_i^2(\beta_i-b)} \right)\right)^{\frac{1}{m}},
	\end{split}
\end{equation}
where  $| MrC_{f}(t, \omega)| >\gamma$,   the parameter $\gamma>0$ is a hard threshold on $| MrC_{f}(t, \omega)|$  to  overcome the shortcoming  that  $| MrC_{f}(t, \omega)|\approx0$. Obviously,  the set of points $\omega=\phi'(t)$  satisfies
\begin{equation}
	\begin{split}
		MrIF(t, \omega):=	\left| \frac{	(\prod_{i=1}^m C_{f}^{th_{\sigma_i, \beta_i}}(t, \omega))^{\frac{1}{m}}}{	MrC_{f}(t, \omega)}\right| =0.
	\end{split}
\end{equation}
We call equation  (38)  the combined  IF   equation  because it can be seen as a combination of multiple IF equations with respect to  different parameters.  The $MrIF(t, \omega)$  is  derived   from  multi-resolution chirplet representation,  which  provides a more accurate IF estimation due to  the high-resolution of original transform.  More importantly,  the    transient signal satisfies the  combined  IF   equation, as stated  in the following theorem.  

\begin{theorem}
	For impulse signal $f(t)=A_0\delta(t-t_0)$, its group delay $t=t_0$ satisfies
	\begin{align} 
		MrIF(t_0, \omega)=0. 
	\end{align} 
\end{theorem}
\begin{proof}
	See the  Supplementary Materials-\uppercase\expandafter{\romannumeral3}.
\end{proof}
%

Theorem 5 shows a very interesting result that the combined IF equation not only can  cover the  instantaneous characteristic of chirp-like  signals  but also can  cover the instantaneous information of impulse signals,  which  unifies  the well-known  IF and  group delay detectors  proposed  in  the SST [21,26] and the time-reassigned SST [31,32].

Based on  the combined   IF equation (38), a novel TF representation called the multi-resolution synchroextracting chirplet transform (MrSECT) is proposed, which is
\begin{gather}
	MrSEC(t,\omega)=	MrC_{f}(t, \omega) \delta(MrIF(t, \omega)).
\end{gather}
Considering  the  calculation  error, it is suggested that (40) be expressed  as
\begin{equation}
	\begin{split}
	&MrSEC(t,\omega)\\=& \begin{cases}
		MrC_{f}(t, \omega),   & \text{$ \left| \frac{	(\prod_{i=1}^m C_{f}^{th_{\sigma_i, \beta_i}}(t, \omega))^{\frac{1}{m}}}{	MrC_{f}(t, \omega)}\right|<\frac{\Delta\omega}{2}$},\\
		0,  & \text{otherwise,}
	\end{cases}
	\end{split}
\end{equation}
where $\Delta\omega$ is the discrete frequency interval.

\section{Numerical validation}

To illustrate the proposed methods, we employ several examples involving  multi-component chirp signals, cosine modulated signals  and impulsive signals which  with fast-varying frequency characteristics.  For comparison,  the STFT with different window width, the adaptive CTs (i.e., double-adaptive CT (DACT) [54] and GLCT [56]),  and the  combined multi-resolution spectrogram (MrSTFT for short) [50] have also been considered.  In addition, the state-of-the-art highly concentrated schemes, namely the  SST [21], RM [18], time-reassigned SST (TSST for short) [31],  second-order SST (SSST for short) [26], and SECT [30],  have also been used.  The window function used in these methods is unified as a Gaussian window, in which the parameter $\sigma $  is  adjusted manually with the aim of providing the best achievable representations according to the visual inspection of the resulting TF plots.

\subsection{Example 1}
We first  use the proposed  methods  to detect close chirp components, which is given by 
\begin{equation}
	\begin{split}
		&f(t)=f_1(t)+f_2(t)+f_3(t)+f_4(t)+n(t), \\
		&f_1(t)=
		\begin{cases}
			\sin(2\pi\times(20t+40t^2)),   &  ~ t\in [0, 1],\\
			0,  & ~ t\in (1, 2], 
		\end{cases}\\
		&f_2(t)=
		\begin{cases}
			\sin(2\pi\times(34t+40t^2)),   &  ~ t\in [0.1, 0.78],\\
			0,  & ~ t\in [0, 0.1) \cup (0.78, 2], 
		\end{cases}\\
		&f_3(t)=FT(0.02\exp(2\pi j(430t-5t^2))), ~  t\in [0, 2],\\
		&f_4(t)=FT(0.02\exp(2\pi j(445t-5t^2))), ~  t\in [0, 2], 
	\end{split}
\end{equation}
where $n(t)$ is the Gaussian  noise  with  the $\text{SNR}=10$ dB, and $FT$ denotes the Fourier transform.   The sampling frequency is  256  Hz.

To illustrate the resolution of MrCT, we first compare it  to the STFT with different window width (i.e., $\sigma=0.1, 0.16, 0.25$) to analyze signals $f_1(t)$ and $f_4(t)$. The TF results are displayed in Fig. 6.  From the results, we can observe that the  STFT using  a smaller value of $\sigma$ (i.e., $\sigma=0.1$)  can  achieve  a highly concentrated  characterization for  $f_4(t)$, but spreads the energy distribution of $f_1(t)$. In contrast, a relatively  large $\sigma$ brings  the opposite   results for STFT when used to address  these  two signals, as presented in Fig. 6(c). Finding  the optimal value of $\sigma$ is important for STFT, particularly when dealing with  a mixture of different types of signals. Unfortunately, it is  not easy  to determine the best one in  practice. In this regard, the multi-resolution representations reduce the difficulty of tuning the parameters and can provide a sharped TF energy distribution (see Fig. 6(d)).  To  further show the  energy concentration more clearly,  we list the TF energy curves  by  selecting  two time and frequency frames (see Fig. 7). The results show that:  the frequency peak is better defined by the MrCT than that by the spectrograms computed with  various analysis windows, which implies the effectiveness of  MrCT to  introduce the  CR parameter, besides the window width parameter; on the other hand, as for signal $f_4(t)$, seen in the second figure, the spectrogram with a short window provides the best representation for this signal, and MrCT is better than other two. Overall, the combination provides a better  signal  representation than any of the spectrograms considering both time and frequency resolutions.
\begin{figure}[hbt!]
	\vspace{-0.1cm}
	\setlength{\belowcaptionskip}{-0.1cm}
	\centering
	\begin{minipage}{0.492\linewidth}
		\centerline{\includegraphics[width=1\textwidth]{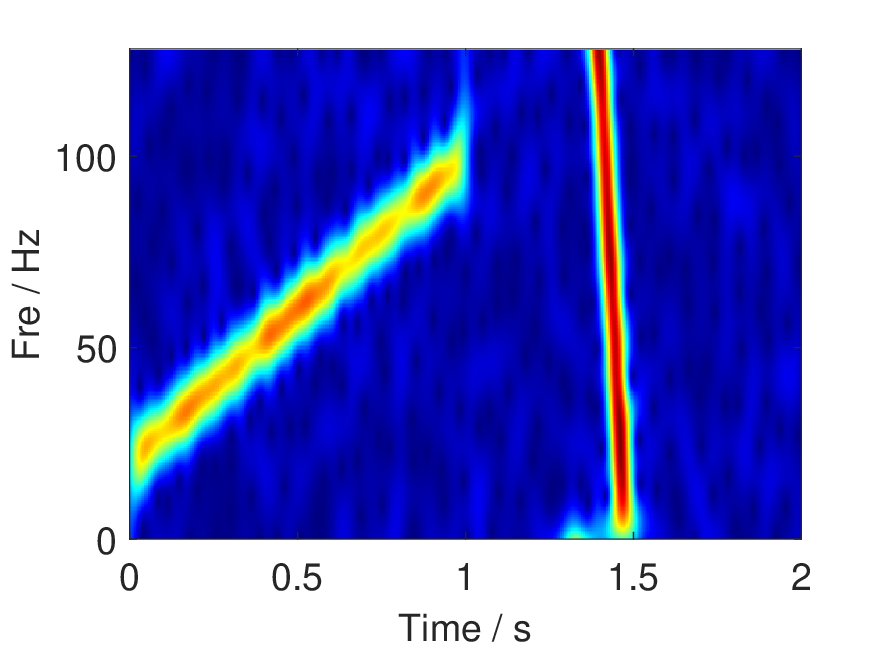}}
		\centerline{(a)}
	\end{minipage}
	\begin{minipage}{0.492\linewidth}
		\centerline{\includegraphics[width=1\textwidth]{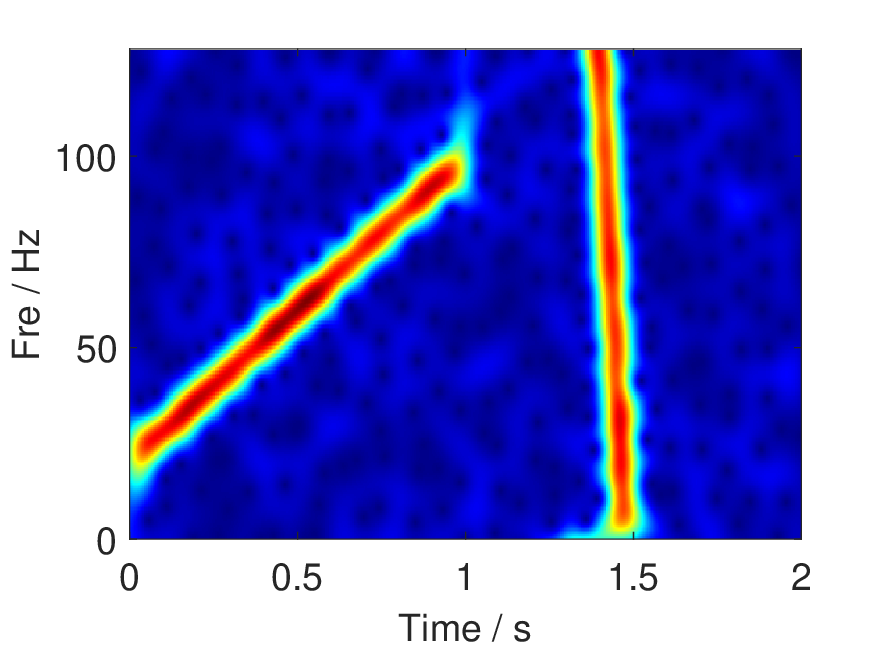}}
		\centerline{(b)}
	\end{minipage}\\ \begin{minipage}{0.492\linewidth}
		\centerline{\includegraphics[width=1\textwidth]{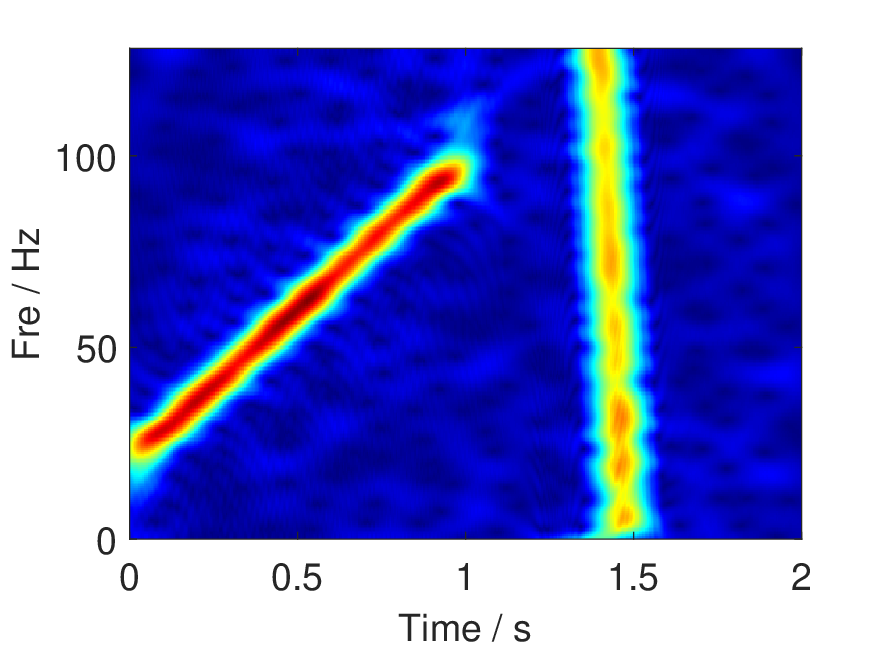}}
		\centerline{(c)}
	\end{minipage}
	\begin{minipage}{0.492\linewidth}
		\centerline{\includegraphics[width=1\textwidth]{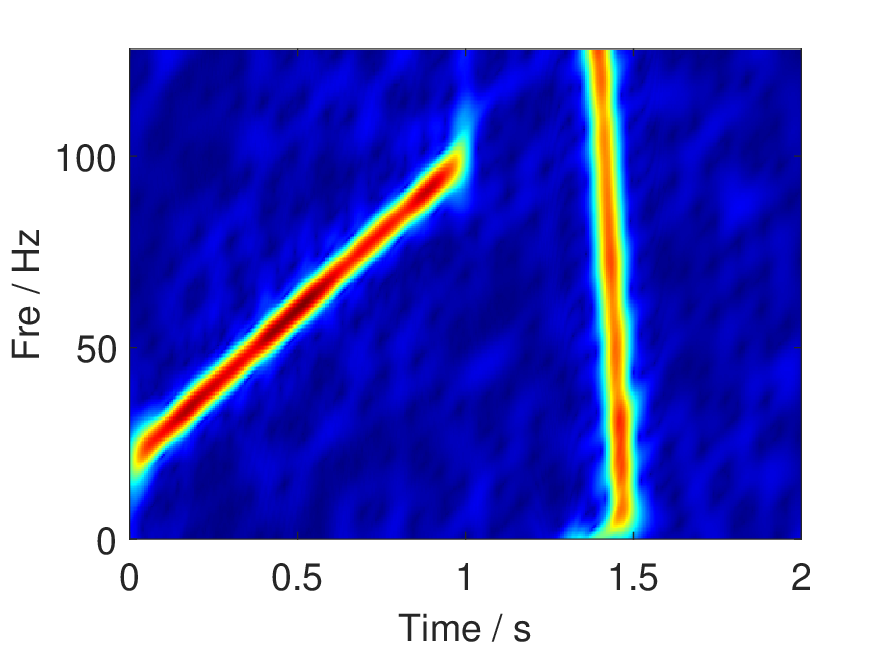}}
		\centerline{(d)}
	\end{minipage}
	\caption{ \small  TF results using STFT and MrCT methods.  (a)  The STFT  with  $\sigma=0.1$,  (b)  the STFT  with  $\sigma=0.16$, (c)  the STFT  with  $\sigma=0.25$, (d) the MrCT ($m=3$, $\sigma=0.1, 0.16, 0.25$).}
\end{figure}
\begin{figure}[htbt!]
	\vspace{-0.1cm}
\setlength{\belowcaptionskip}{-0.1cm}
	\centering
	\begin{minipage}{0.493\linewidth}
		\centerline{\includegraphics[width=1\textwidth]{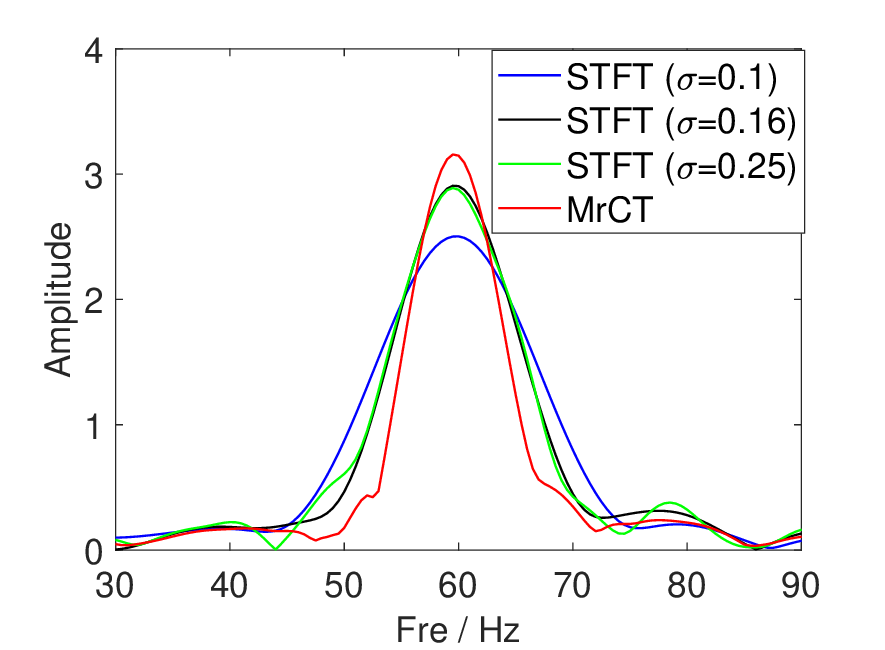}}
		\centerline{(a)}
	\end{minipage}
	\begin{minipage}{0.493\linewidth}
		\centerline{\includegraphics[width=1\textwidth]{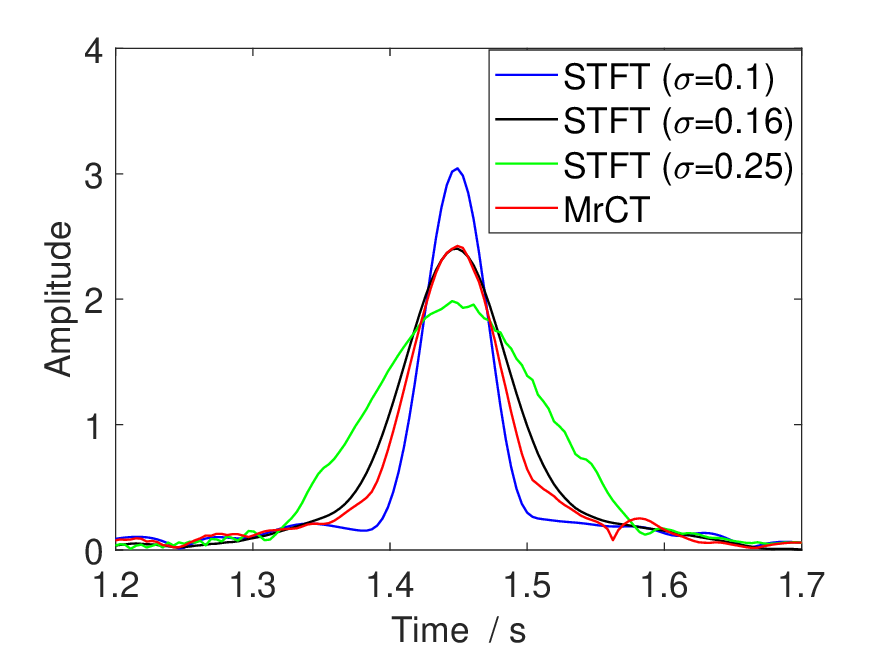}}
		\centerline{(b)}
	\end{minipage}
	\caption{ \small  TF energy  curves by selecting  various frames.  (a) Spectrum of the time frame $t=0.5$ s, (b) time evolution of the frequency bin Fre $=46$  Hz.}
\end{figure}
\begin{figure}[ht!]
	\vspace{-0.1cm}
\setlength{\belowcaptionskip}{-0.2cm}
	\centering
	\begin{minipage}{0.492\linewidth}
		\centerline{\includegraphics[width=1\textwidth]{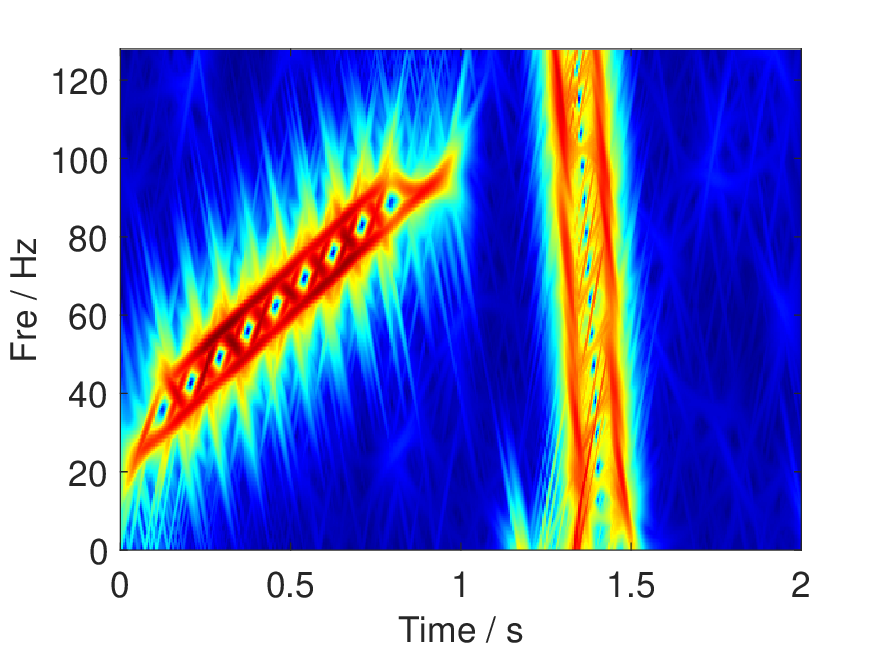}}
		\centerline{(a)}
	\end{minipage}
	\begin{minipage}{0.492\linewidth}
		\centerline{\includegraphics[width=1\textwidth]{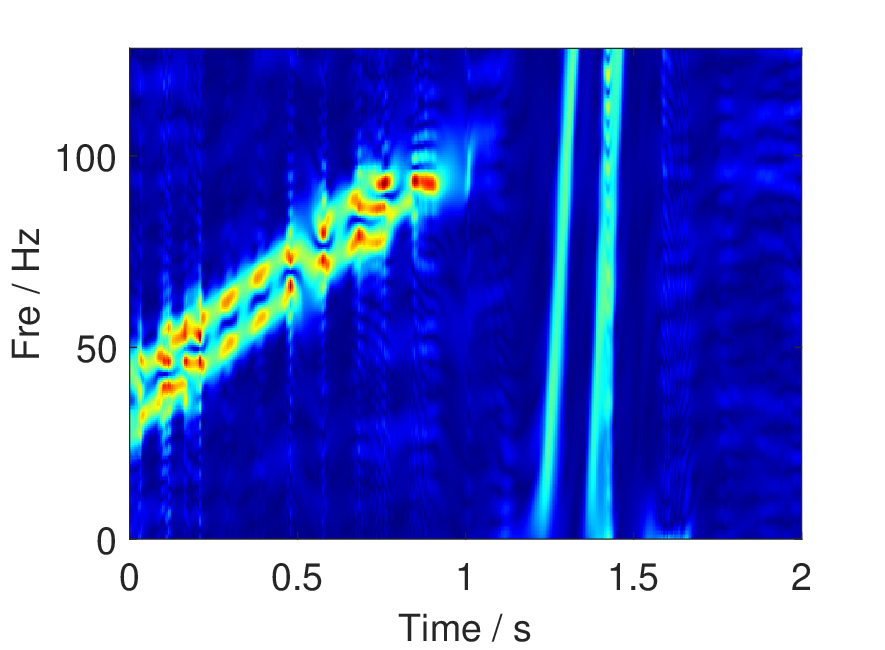}}
		\centerline{(b)}
	\end{minipage}\\ \begin{minipage}{0.492\linewidth}
		\centerline{\includegraphics[width=1\textwidth]{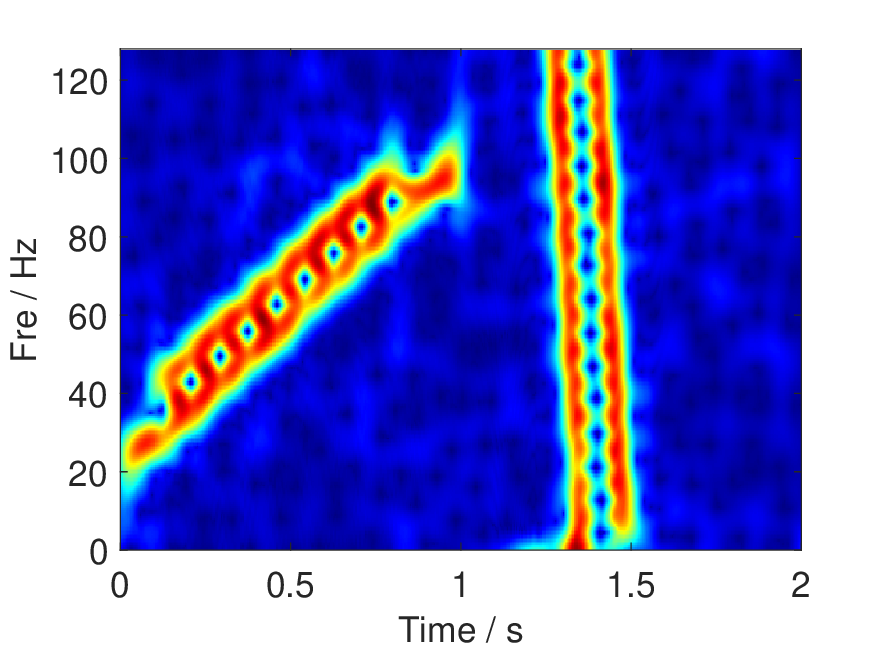}}
		\centerline{(c)}
	\end{minipage}
	\begin{minipage}{0.492\linewidth}
		\centerline{\includegraphics[width=1\textwidth]{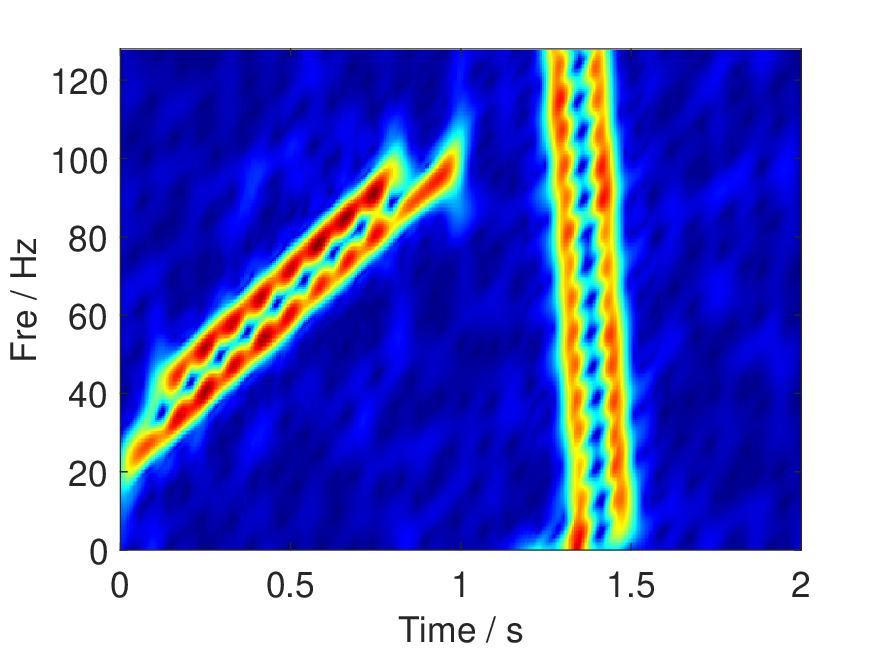}}
		\centerline{(d)}
	\end{minipage}
	\caption{ \small  TF results of signal (42) using various analysis  methods.  (a) GLCT, (b) DACT, (c) MrSTFT, (d)  MrCT.}
\end{figure}

\begin{figure}[hbt!]
	\vspace{-0.2cm}
	\setlength{\belowcaptionskip}{-0.2cm}
	\centering
	\begin{minipage}{0.48\linewidth}
		\centerline{\includegraphics[width=1\textwidth]{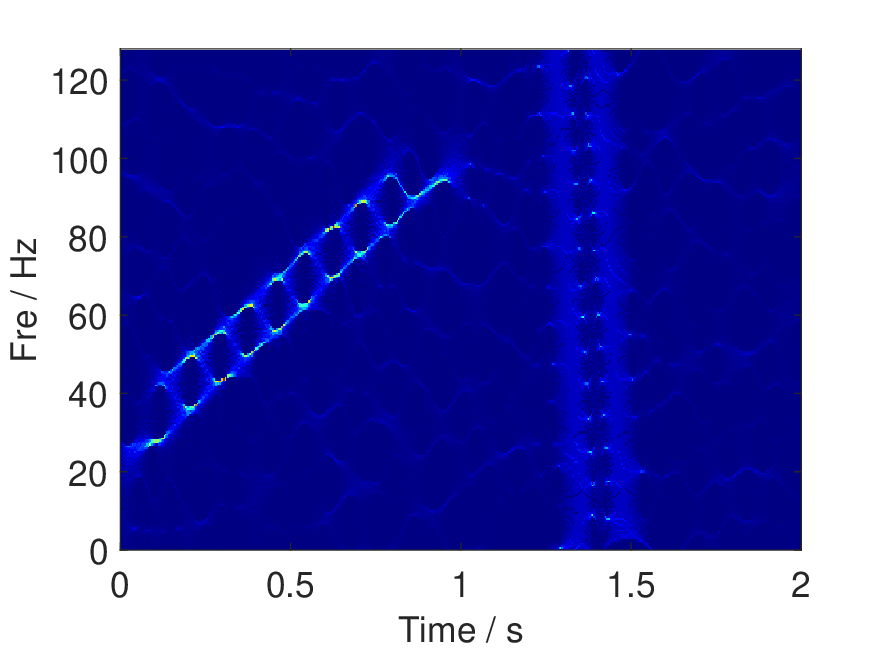}}
		\centerline{(a)}
	\end{minipage}
	\begin{minipage}{0.48\linewidth}
		\centerline{\includegraphics[width=1\textwidth]{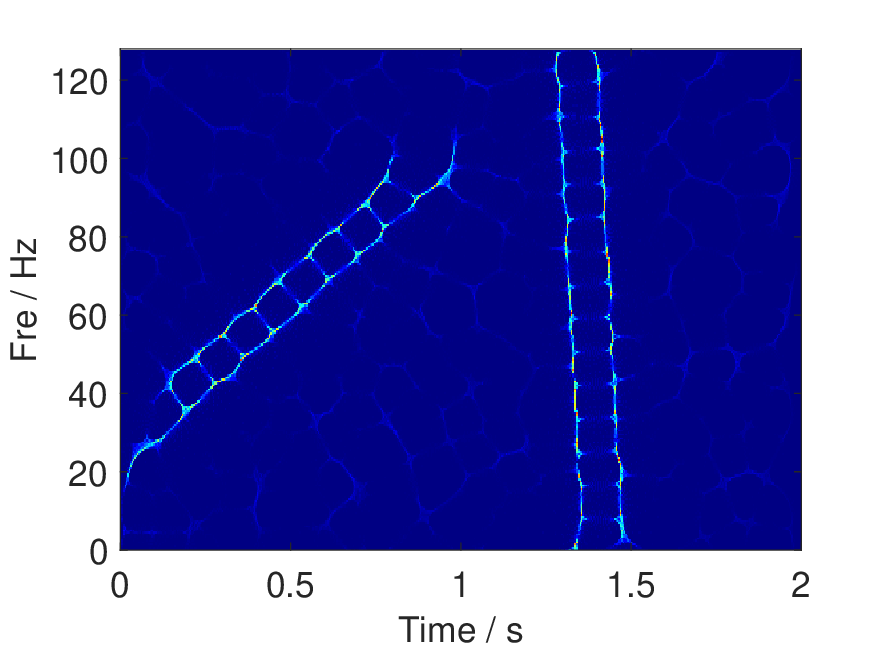}}
		\centerline{(b)}
	\end{minipage} \begin{minipage}{0.48\linewidth}
		\centerline{\includegraphics[width=1\textwidth]{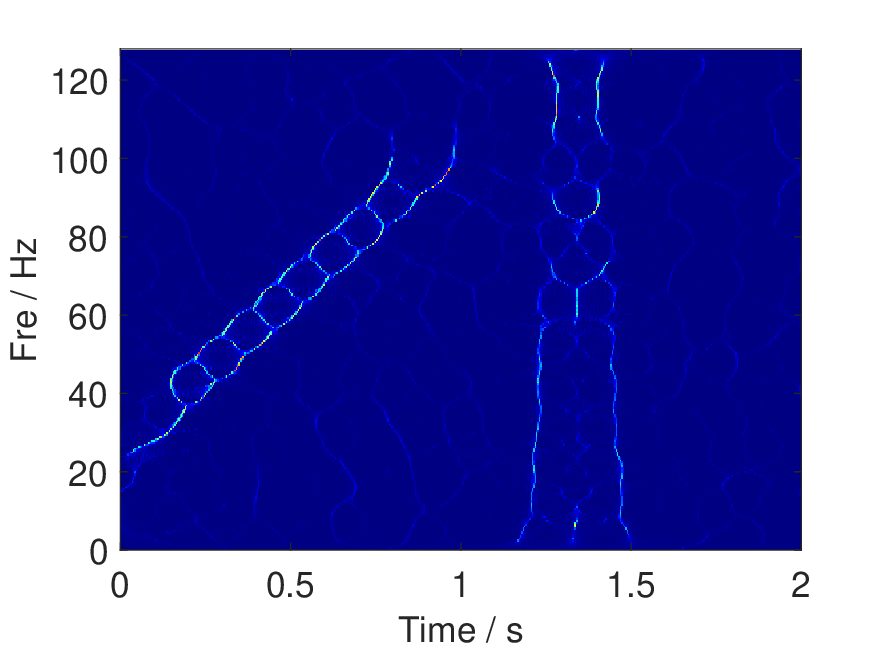}}
		\centerline{(c)}
	\end{minipage}
	\begin{minipage}{0.48\linewidth}
		\centerline{\includegraphics[width=1\textwidth]{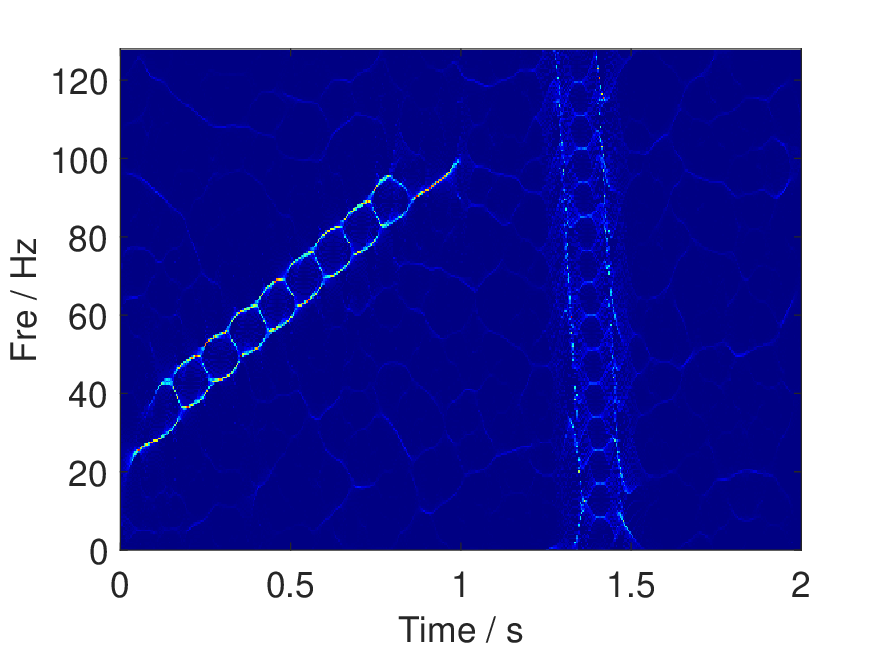}}
		\centerline{(d)}
	\end{minipage}
	\begin{minipage}{0.48\linewidth}
		\centerline{\includegraphics[width=1\textwidth]{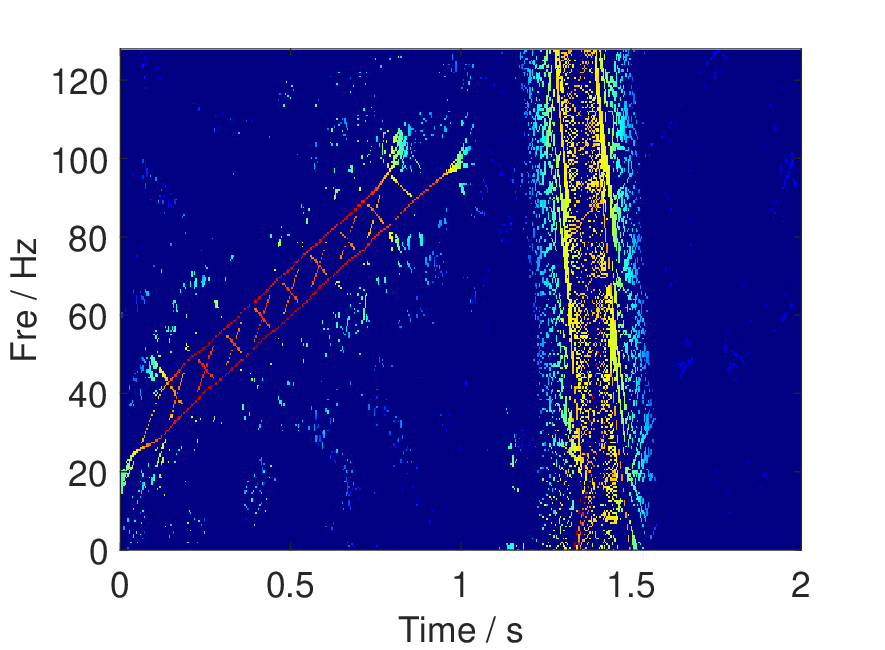}}
		\centerline{(e)}
	\end{minipage}
	\begin{minipage}{0.48\linewidth}
		\centerline{\includegraphics[width=1\textwidth]{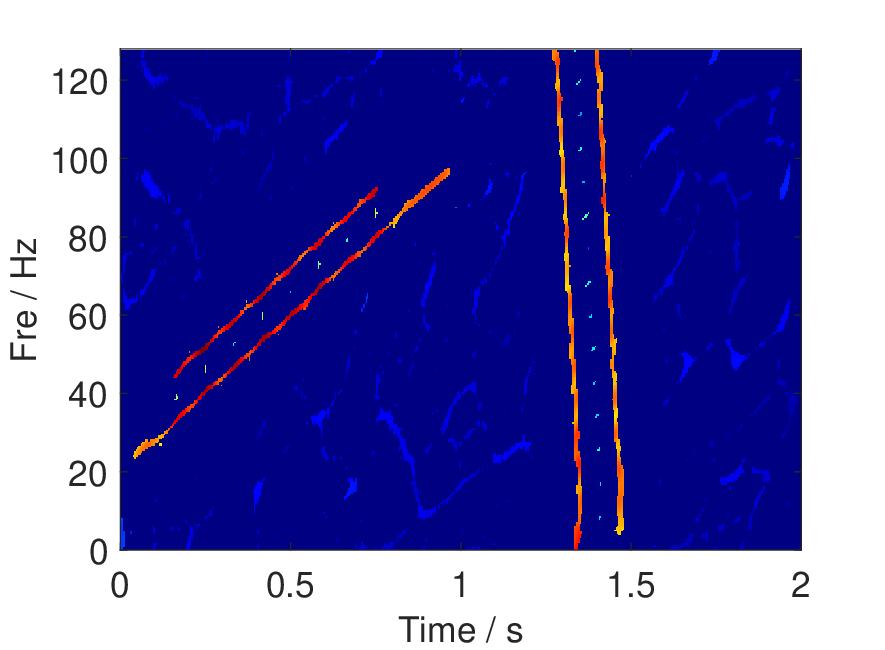}}
		\centerline{(f)}
	\end{minipage}
	\caption{ \small  TF results by various post-processing methods.  (a)  TF plot by SST, (b) TF plot by RM,  (c)  TF plot by TSST, (d) TF plot by SSST,   (e)  TF plot by SECT, (f) TF plot by  MrSECT.}
\end{figure}

Fig. 8  presents the TF representations by two adaptive CTs (i.e., GLCT and DACT) and two combined methods (i.e., MrSTFT and  MrCT, here these two combinations use the same values of $\sigma$).   From  the results, we can see that the MrCT yields a high-concentration  TF result,  and the close components are clearly separated. While the other three methods lead to the TF plots being difficult to interpret.  In addition,   the computational times  required for the  GLCT, DACT,  MrSTFT, and MrCT   in addressing  signal (42) are  0.312  s,   0.031 s,  0.027 s, and  0.041 s, respectively (the tested computer configuration: Intel Core i5-6600 3.30 GHz, 16.0 GB of RAM,  and MATLAB version R2020b). The last three methods need  much less calculation comparing  with  the  GLCT in which the window function is adjusted at every TF location.  

Fig. 9 displays the TF representations obtained by  six TF post-processing methods, including the SST, RM, TSST, SSST, SECT, and MrSECT. It can be seen from the results that,  there  are  heavy  cross-terms between the first two  modes in the TF plane  for the first five  methods.  The MrSECT is clearly superior to others, yielding high concentration of both close components. Compared to the TF result of MrCT (see Fig. 8(d)),  MrSECT effectively improves  its TF readability.

\subsection{Example 2}
Next, let us consider a four-component AM-FM signal with two cosine modulated modes  and  two impulse components   as
\begin{equation}
	\begin{split}
		&f(t)=f_1(t)+f_2(t)+f_3(t)+f_4(t)+n(t), \\
		&f_1(t)=(1+0.05\cos(20\pi t))\cos(2\pi(9\sin(2\pi t)+80t)),\\
		&f_2(t)=(1+0.1\cos(20\pi t))\cos(2\pi(9\sin(2\pi t)+115t)), \\
		&f_3(t)=5\exp(-10000\pi(t-0.46)^2)\cos(340\pi t),\\
		&f_4(t)=5\exp(-10000\pi(t-0.52)^2)\cos(340\pi t), 
	\end{split}
\end{equation}
where  $n(t)$ is the Gaussian noise with the SNR$=8$ dB.   The sampling frequency is 512 Hz, and  time duration is [0 1s].

Fig. 10  provides the TF plots of signal (43)  obtained by the GLCT, DACT,   MrSTFT,  and  MrCT ($m=6$).  It is seen that serious cross-terms appear between  the close components in the TF representations of GLCT and DACT,   and  these two methods  are  unable to fully segregate time and frequency.  By contrast,   MrSTFT and MrCT provide a faithful local representation, with high resolution in both time and frequency.   Compared to MrSTFT, the MrCT has more concentrated
TF energy at the modulation part (see TF plot  in the time interval [0.1, 0.3], [0.7, 0.9]), showing good performance in  handling  cosine modulation signals. 
\begin{figure}[ht!]
	\vspace{-0.2cm}
	\setlength{\belowcaptionskip}{-0.1cm}
	\centering
	\begin{minipage}{0.492\linewidth}
		\centerline{\includegraphics[width=1\textwidth]{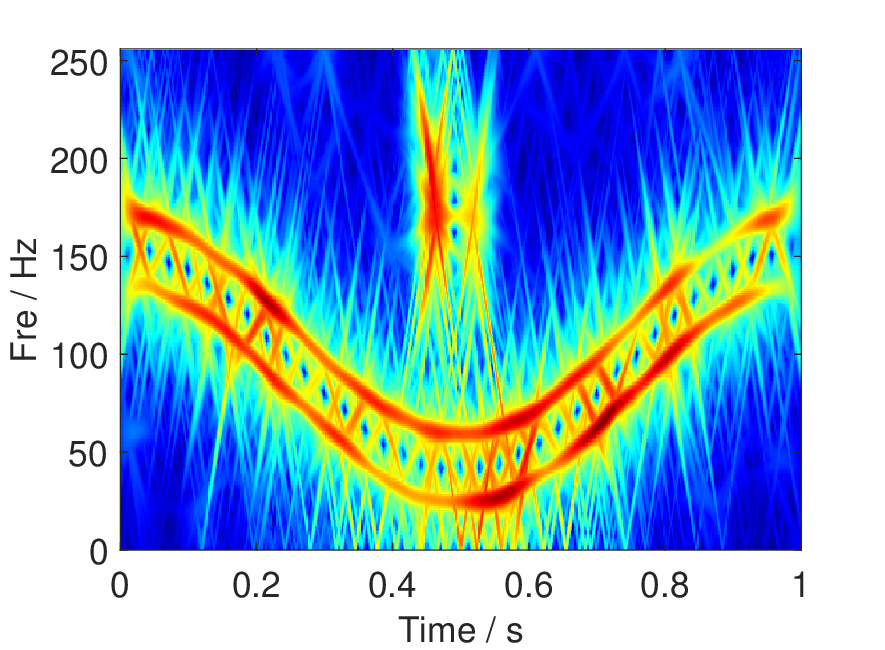}}
		\centerline{(a)}
	\end{minipage}
	\begin{minipage}{0.492\linewidth}
		\centerline{\includegraphics[width=1\textwidth]{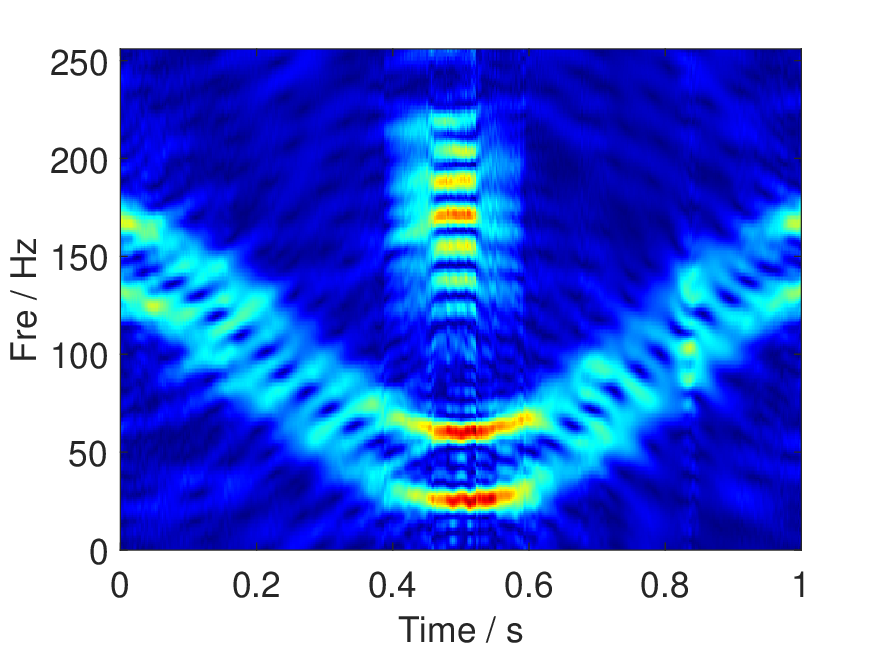}}
		\centerline{(b)}
	\end{minipage}\\ 
\begin{minipage}{0.492\linewidth}
		\centerline{\includegraphics[width=1\textwidth]{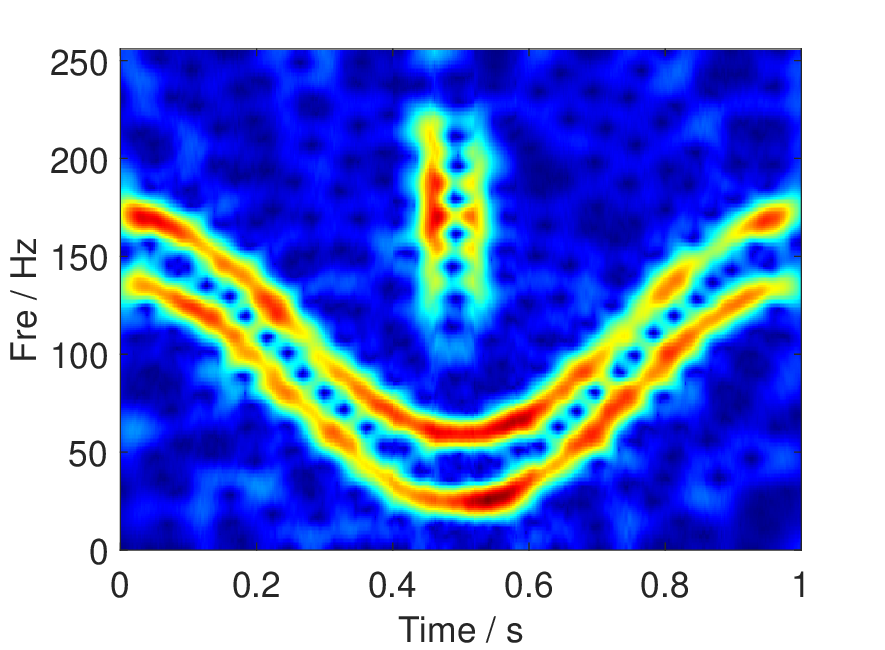}}
		\centerline{(c)}
	\end{minipage}
	\begin{minipage}{0.492\linewidth}
		\centerline{\includegraphics[width=1\textwidth]{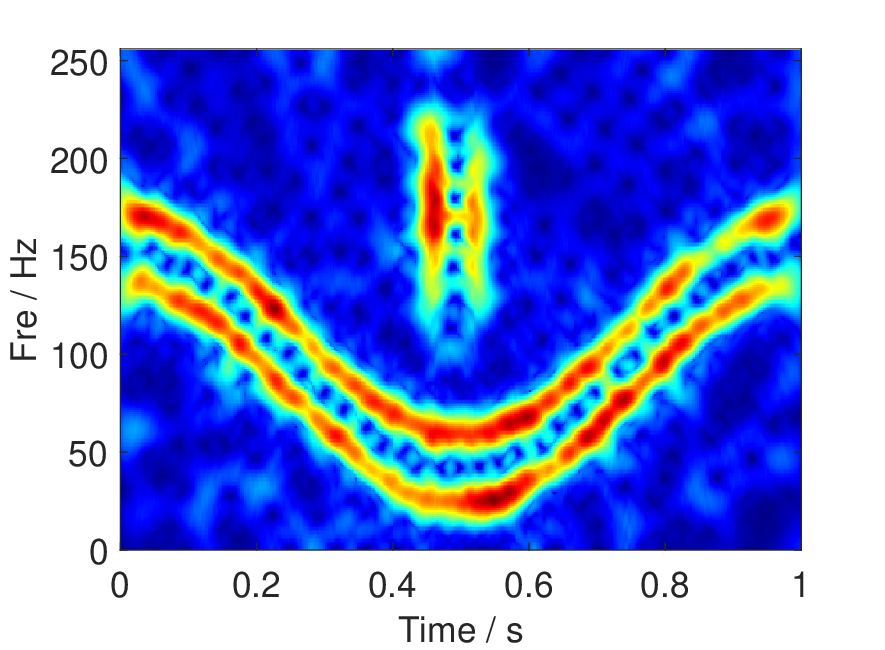}}
		\centerline{(d)}
	\end{minipage}
	\caption{ \small  TF results of signal (43) using various analysis  methods.  (a) GLCT, (b) DACT, (c) MrSTFT, (d)  MrCT.}
\end{figure}

\begin{figure}[ht!]
	\vspace{-0.2cm}
	\setlength{\belowcaptionskip}{-0.1cm}
	\centering
	\begin{minipage}{0.48\linewidth}
		\centerline{\includegraphics[width=1\textwidth]{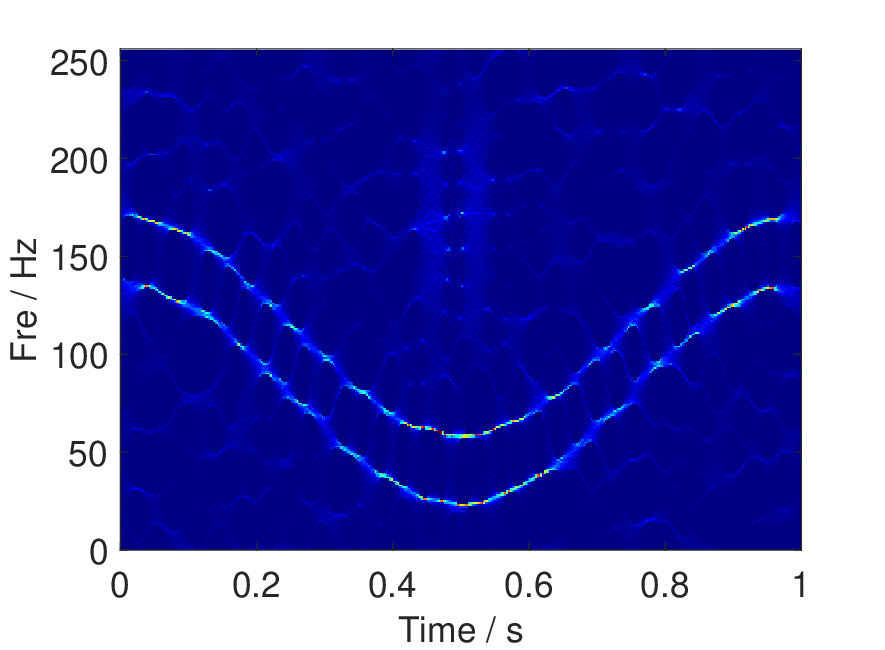}}
		\centerline{(a)}
	\end{minipage}
	\begin{minipage}{0.48\linewidth}
		\centerline{\includegraphics[width=1\textwidth]{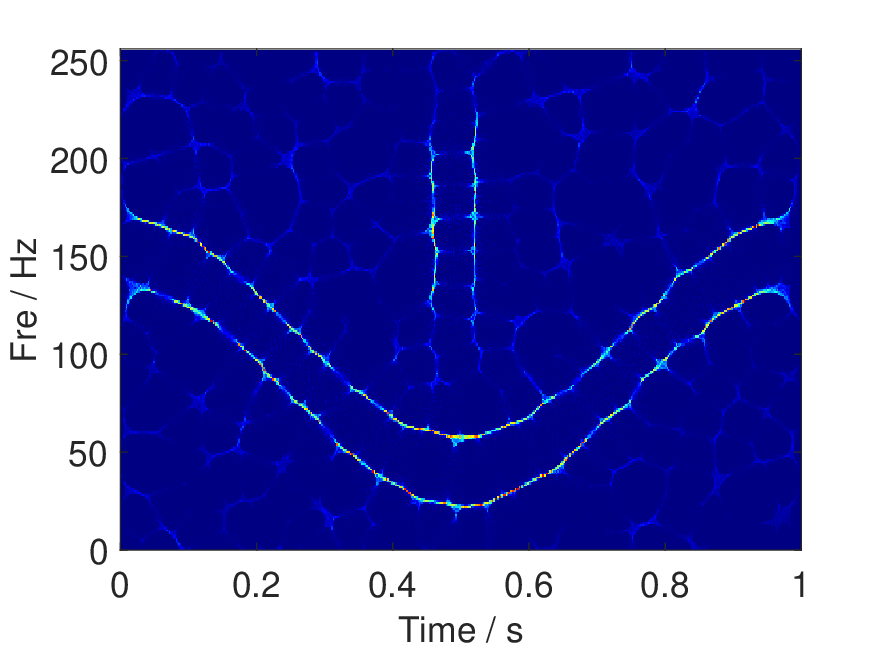}}
		\centerline{(b)}
	\end{minipage} \begin{minipage}{0.48\linewidth}
		\centerline{\includegraphics[width=1\textwidth]{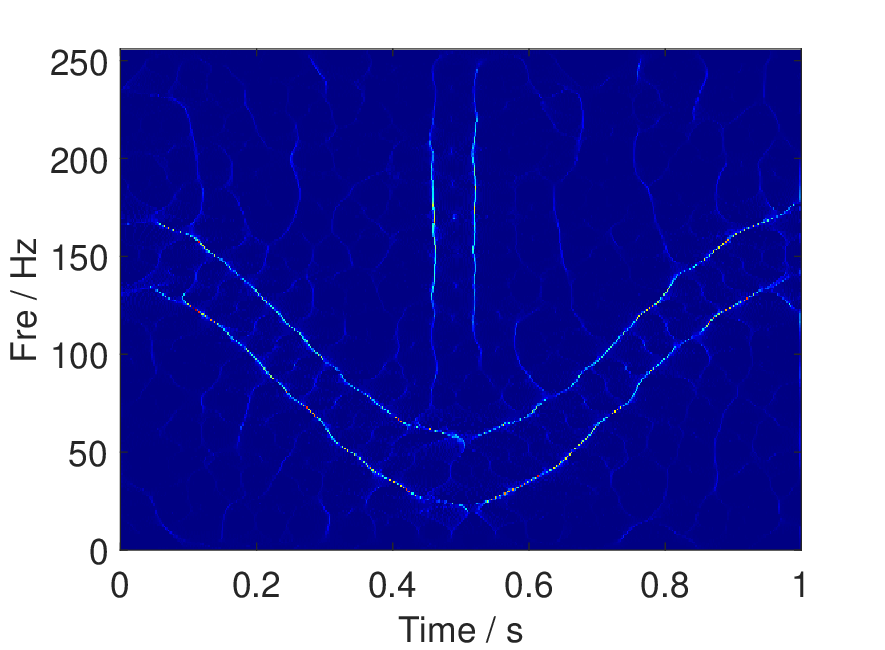}}
		\centerline{(c)}
	\end{minipage}
	\begin{minipage}{0.48\linewidth}
		\centerline{\includegraphics[width=1\textwidth]{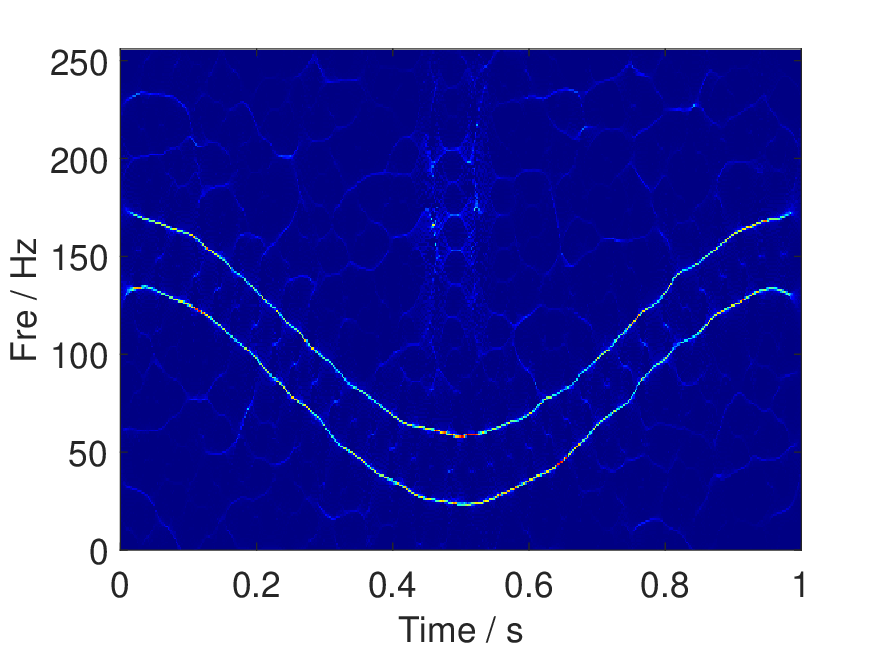}}
		\centerline{(d)}
	\end{minipage}
	\begin{minipage}{0.48\linewidth}
		\centerline{\includegraphics[width=1\textwidth]{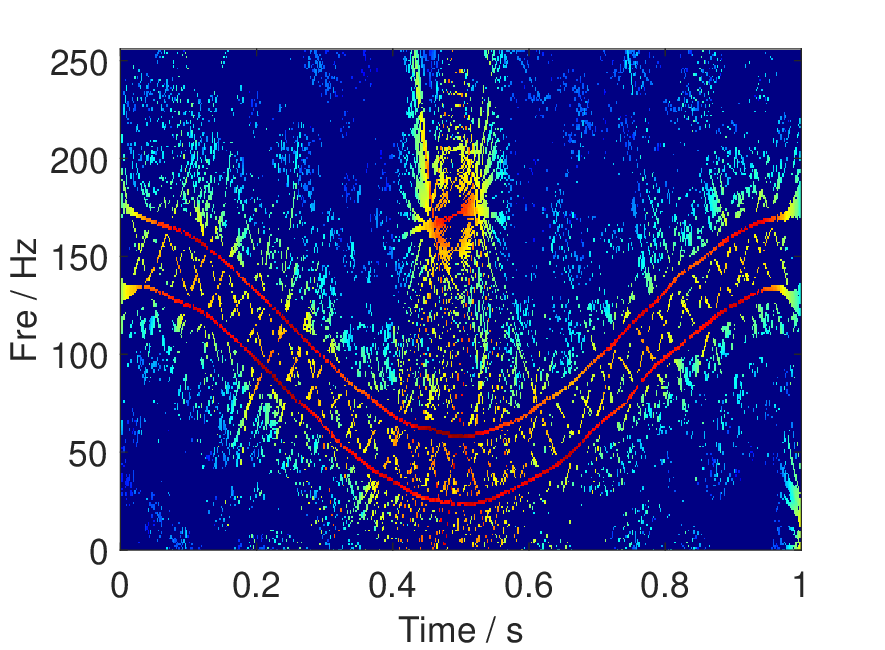}}
		\centerline{(e)}
	\end{minipage}
	\begin{minipage}{0.48\linewidth}
		\centerline{\includegraphics[width=1\textwidth]{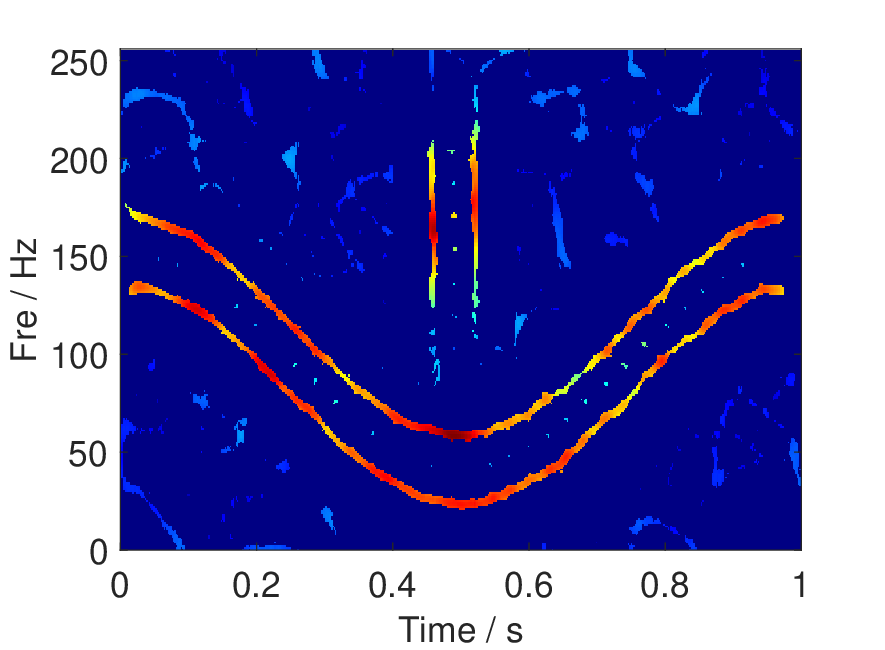}}
		\centerline{(f)}
	\end{minipage}
	\caption{ \small  TF results by various post-processing methods.  (a)  TF plot by SST, (b) TF plot by RM,  (c)  TF plot by TSST, (d) TF plot by SSST,   (e)  TF plot by SECT, (f) TF plot by  MrSECT.}
\end{figure}
Fig. 11 gives the TF representations generated by the six post-processing methods. It clearly shows that the TF result of SST is blurred. The SSST and SECT provide the TF plots  with high energy concentration for $f_1(t)$ and $f_2(t)$,  but fail to characterize impulsive  signals effectively. The TSST  provides  a better TF localization for the transient signals,  while  degrading  the TF sharpness  of the cosine modulation signals  at the stationary part (see TF plot in the time interval [0.4, 0.6]). The RM  is to reassign the spectrogram from the TF direction such that the corresponding time resolution is more concentrated than SST-based methods, but it is not as clear as the MrSECT.  By comparison,  the MrSECT has  the best   TF readability in handling signal (43).  It should be pointed out that the proposed MrSECT is subject to the boundary effects, missing a small amount of boundary information when pursuing a high-concentration TF representation.

\subsection{Application }
A typical biomedical signal, ECG signal from the BIDMC
the Dataset [71], is employed to validate our proposed method.
The sample rate of the ECG signal is 125 Hz. About 8 s ECG segment and its STFT are displayed in Fig. 12.
\begin{figure}[ht!]
	\vspace{-0.2cm}
	\setlength{\belowcaptionskip}{-0.2cm}
	\centering
	\begin{minipage}{0.493\linewidth}
		\centerline{\includegraphics[width=1\textwidth]{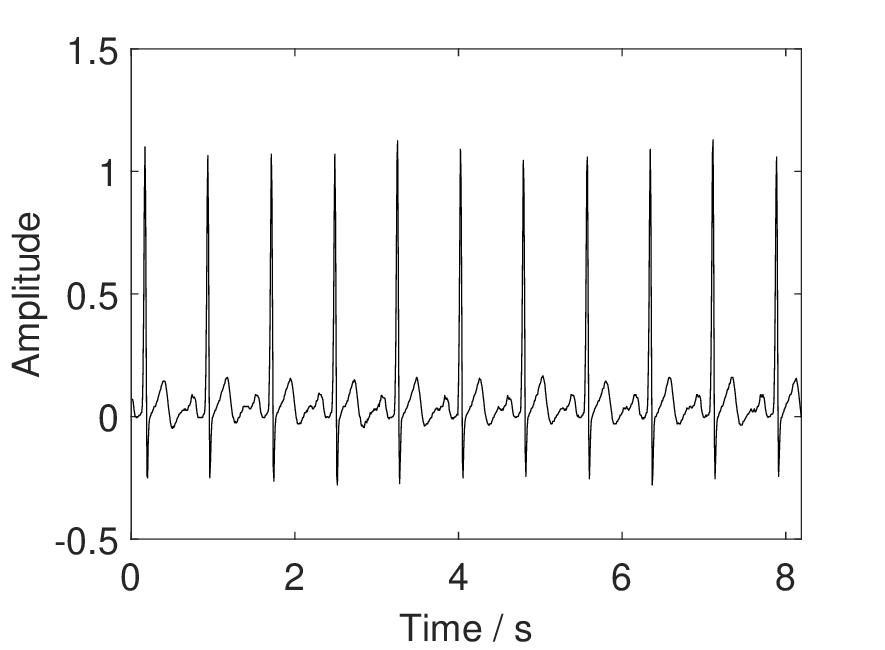}}
		\centerline{(a)}
	\end{minipage}
	\begin{minipage}{0.493\linewidth}
		\centerline{\includegraphics[width=1\textwidth]{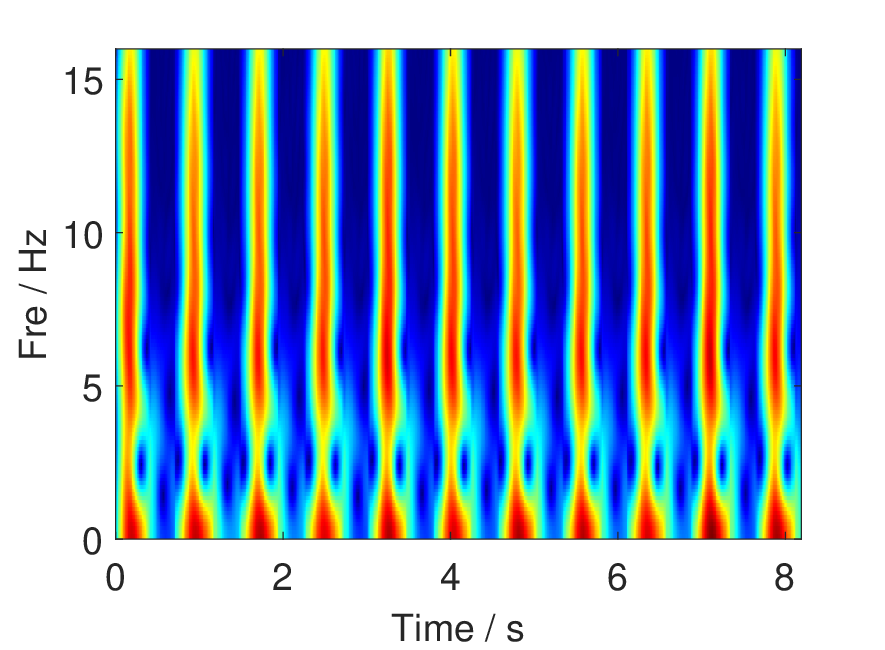}}
		\centerline{(b)}
	\end{minipage}
	\caption{ \small  The time series of the ECG signal and its spectrogram.}
\end{figure}
\begin{figure}[ht!]
	\vspace{-0.2cm}
	\setlength{\belowcaptionskip}{-0.2cm}
	\centering
	\begin{minipage}{0.48\linewidth}
		\centerline{\includegraphics[width=1\textwidth]{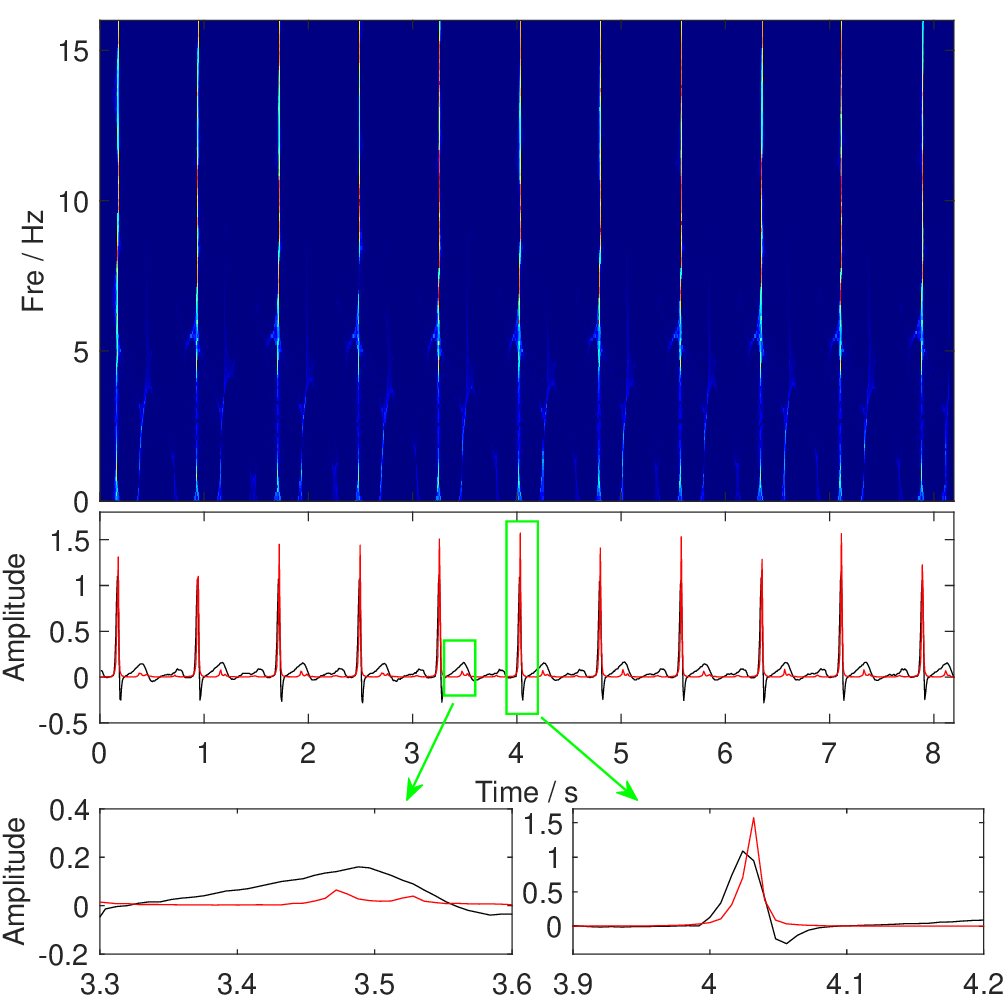}}
		\centerline{(a)}
	\end{minipage}
	\begin{minipage}{0.48\linewidth}
		\centerline{\includegraphics[width=1\textwidth]{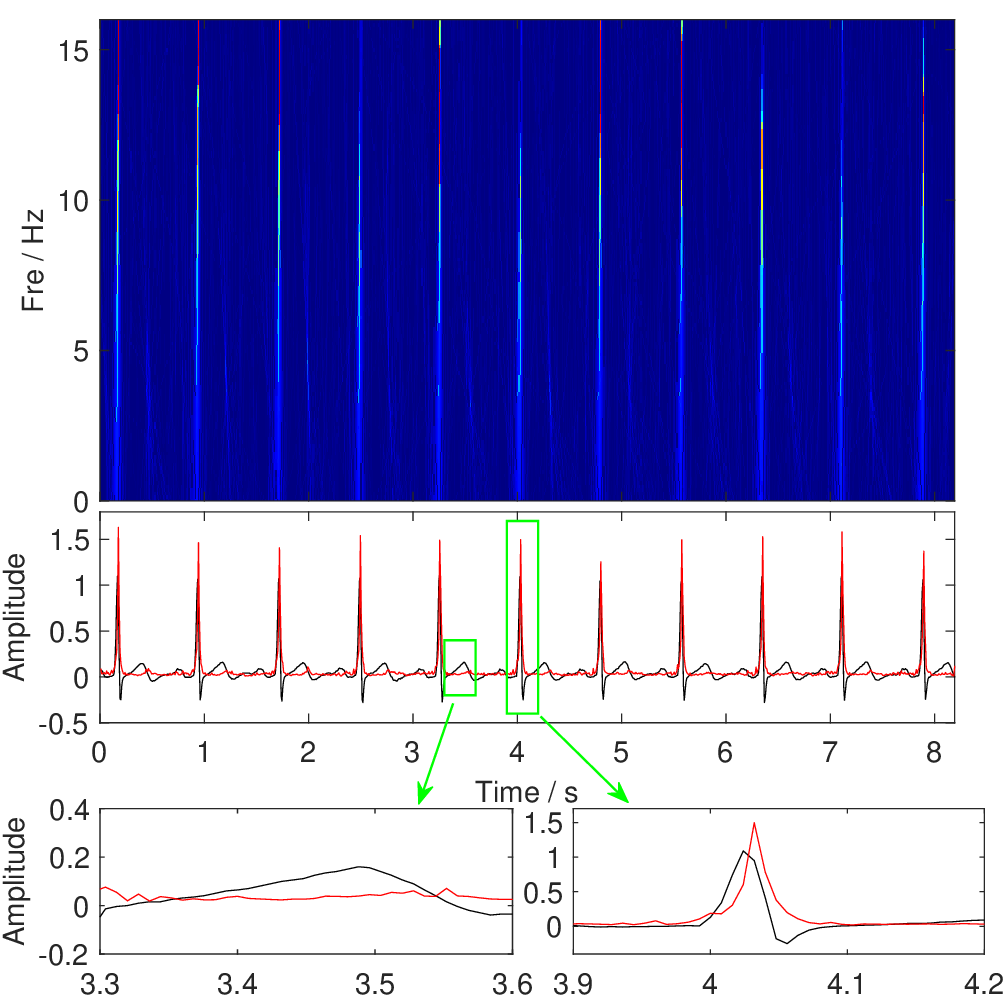}}
		\centerline{(b)}
	\end{minipage}\\ \begin{minipage}{0.48\linewidth}
		\centerline{\includegraphics[width=1\textwidth]{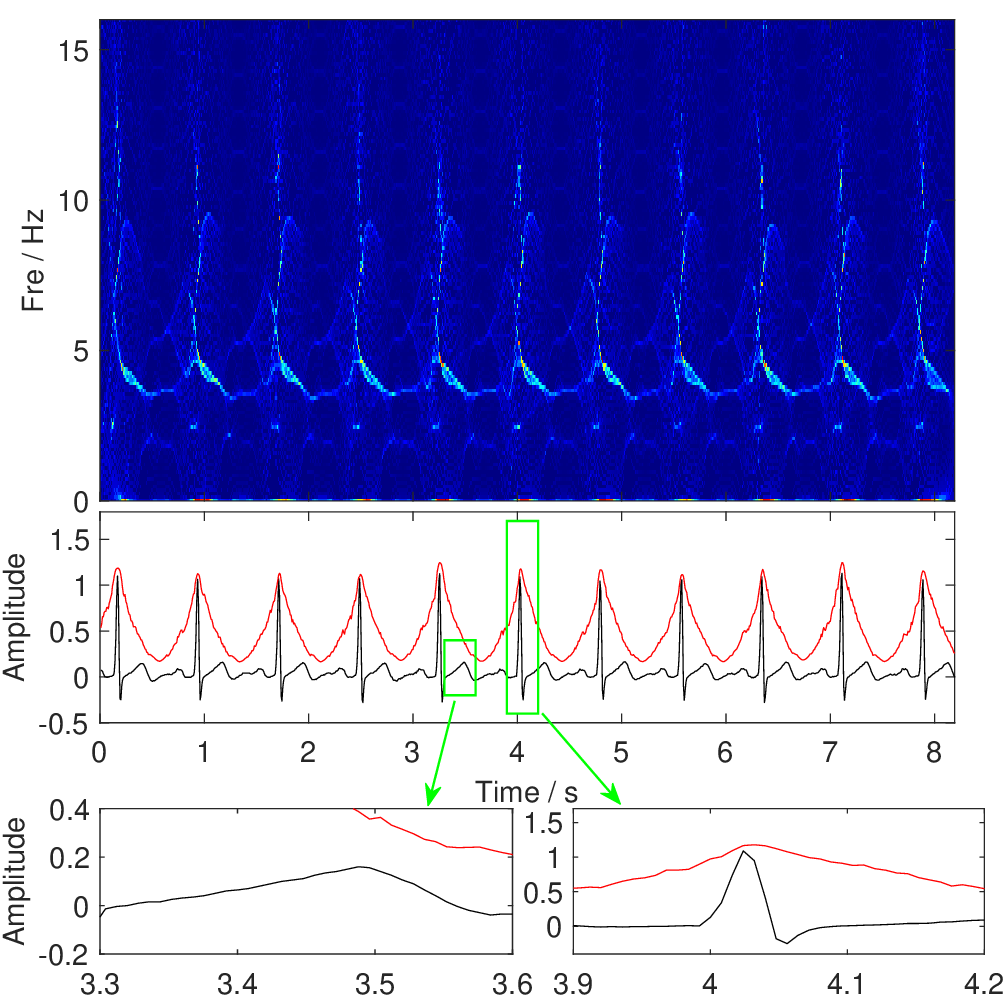}}
		\centerline{(c)}
	\end{minipage}
	\begin{minipage}{0.48\linewidth}
		\centerline{\includegraphics[width=1\textwidth]{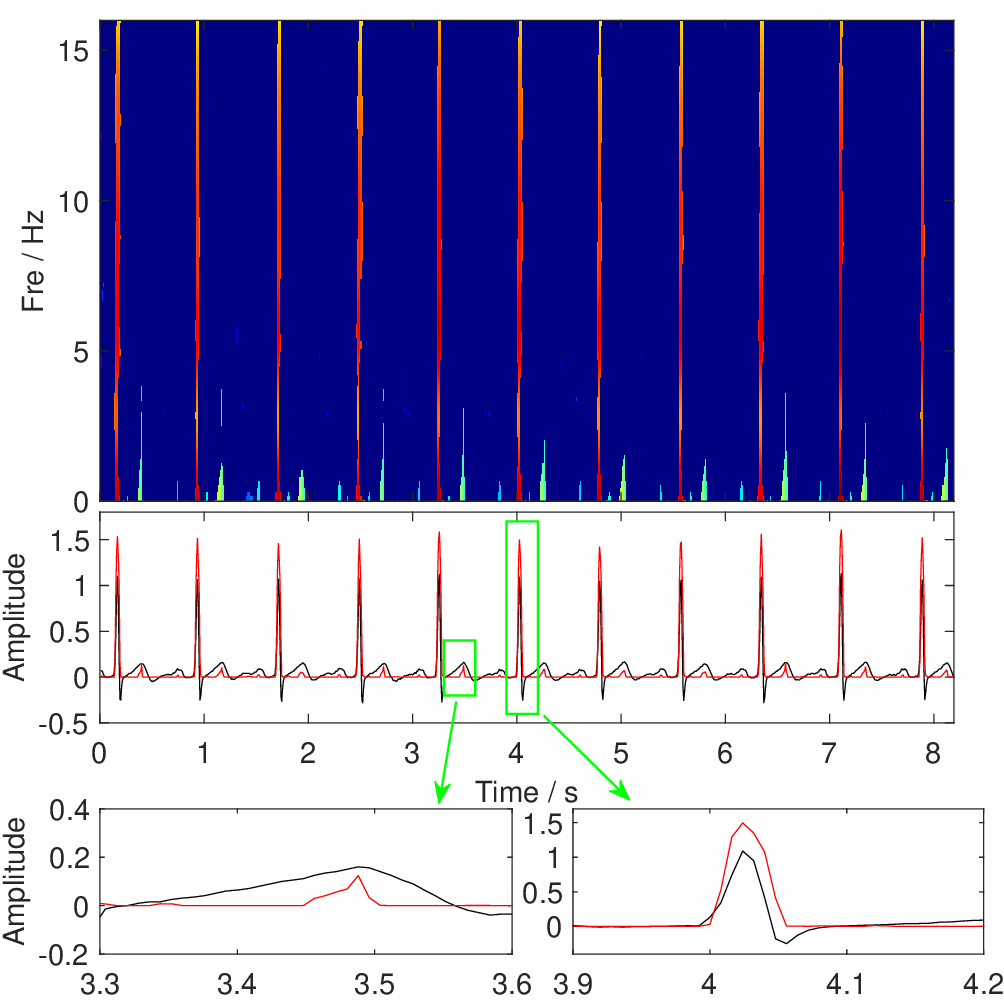}}
		\centerline{(d)}
	\end{minipage}
	\caption{ \small The ECG's analysis results calculated by different methods: (a) RM, (b) TSST, (c) SSST, (d) MrSECT. In the below of each TF result, there are time distribution of each TF representation (red -), the original signal (blak -), and locally-zoomed results, as marked with the green boxes.}
\end{figure}

Fig. 13 shows the analysis results of ECG signal estimated by four different methods: (a) RM, (b) TSST, (c) SSST, (d)
MrSECT. The SSST presents a blurry TF representation, based on which, it is difficult to make further analysis and  judgment. Fig. 13(a,b,d) indicate that RM, TSST, and MrSECT achieve good energy concentration for the spikes of the ECG signal. These  methods (especially the MrSECT) can resolve the spikes and describe the transient characteristics clearly. We also give the time-domain marginal distribution of each TF result as well as  locally-zoomed results for the parts marked with the green boxes. As illustrated  in Fig. 13,   the result obtained by our MrSECT approach  achieves superior effectiveness in TF dynamic estimation because it reflects the dynamic change of the ECG signal more accurately.

\section{Conclusion}
In this paper, we focused on improving the TF resolution of the standard chirplet transform (CT) for multi-component
non-stationary signal analysis. We  theoretically analyzed the effect of the chirp-based Gaussian window on the CT and compared the CT with  the rotation-window CT.  The given parameters analysis  showed some interesting results like that a narrow window limits the matching capacity of chirp basis,  and provided  theoretical   instruction  for CT to choose suitable parameters  to  obtain a high-resolution TF representation. To overcome the limitation of CT in dealing with  signals consisting of a mixture of chirps and pulses,   we   proposed the  multi-resolution chirplet transform (MrCT) by computing the geometric mean of multiple CTs, which takes advantage of multiple estimates at a range of temporal resolutions and frequency bandwidths and localizes the signal in both time and frequency better than it is possible with any single CT.  Moreover, based on the  combined IF equation, we developed  the multi-resolution synchroextracting chirplet transform (MrSECT) to further improve the readability of the proposed  MrCT. The numerical  results  and real signal application  confirmed  the effectiveness  of the proposed approaches.  

Although the  MrCT and MrSECT  have a good performance in addressing  closely-spaced signals, some works should be considered.
\begin{itemize}
	\setlength{\itemsep}{0pt}
	\item [1)] The noise robustness analysis of the proposed MrCT approach, as well as a deeper understanding of the influence of noise on the  combined IF equation,  should  be discussed. 
	\item [2)] The MrSECT is subject to the boundary effects, how to minimize the boundary effects based on forecasting techniques is worthy of consideration.
	\item [3)]The application to other real-life signals (e.g.,  radar signal [4], mechanical vibration signal [32]) needs  to  be further investigated.
\end{itemize}

\section*{Acknowledgements}
We thank K. Abratkiewicz  for  sharing the implementation of  his  work. 


%


\end{document}